\documentclass[a4paper,numberwithinsect, 11pt]{article}

\usepackage[margin=1in]{geometry}
\usepackage{hyperref}
\usepackage{breakurl}

\usepackage[USenglish]{babel}

\usepackage{microtype}%

\usepackage{nicefrac}

\usepackage{enumitem}

\usepackage{algorithm}
\usepackage[noend]{algpseudocode}

\usepackage{xspace, xcolor}

\usepackage{amsmath,amssymb,amsthm, mathtools}

\usepackage{complexity}

\theoremstyle{plain}
\newtheorem{thm}{Theorem}[section]
\newtheorem{lem}[thm]{Lemma}

\newtheorem{cor}[thm]{Corollary}
\newtheorem{defn}[thm]{Definition}
\newtheorem{prop}[thm]{Proposition}

\newtheorem{claim}[thm]{Claim}
\newtheorem{problem}[thm]{Problem}
\newtheorem{hypo}[thm]{Hypothesis}

\newcommand{\nats}{\mathbb{N}}

\newcommand{\Oh}{{\cal O}}

\newcommand{\NAND}{\mathrm{NAND}}
\newcommand{\Impl}{\mathrm{IMPL}}

\newcommand{\confam}{\mathcal{F}}
\newcommand{\Implication}{\textsc{Implications}\xspace}
\newcommand{\wDAGImplication}{\textsc{Weighted DAG Implications}\xspace}
\newcommand{\SubsetSum}{\textsc{SubsetSum}\xspace}
\newcommand{\THC}[1]{T_{#1\text{-HC}}}
\newcommand{\vars}{\mathrm{vars}}
\newcommand{\ones}{\mathrm{ones}}

\newcommand{\divides}{\mid}
\newcommand{\notdivides}{\nmid}

\newcommand{\eps}{\epsilon}

\newcommand{\entry}[1]{#1 \dots #1}

\def\cqedsymbol{\ifmmode$\lrcorner$\else{\unskip\nobreak\hfil
\penalty50\hskip1em\null\nobreak\hfil$\lrcorner$
\parfillskip=0pt\finalhyphendemerits=0\endgraf}\fi} 

\newcommand{\cqed}{\renewcommand{\qed}{\cqedsymbol}}

\title{Finding Small Satisfying Assignments Faster Than Brute Force:\newline A Fine-grained Perspective into Boolean Constraint Satisfaction}

\author{Marvin K\"unnemann\thanks{Max Planck Institute for Informatics, Saarland Informatics Campus, Saarbr\"ucken, Germany. \texttt{\{marvin|dmarx\}@mpi-inf.mpg.de}. Research of the second author was supported by funding from the European Research Council (ERC) under the European Union's Horizon 2020 research and innovation programme under grant agreement SYSTEMATICGRAPH (No.~725978).} \and
D\'aniel Marx\footnotemark[1]}%

\date{}

\begin{document}

\maketitle

\begin{abstract}
To study the question under which circumstances small solutions can be found faster than by exhaustive search (and by how much), we study the fine-grained complexity of Boolean constraint satisfaction with size constraint exactly $k$. More precisely, we aim to determine, for any finite constraint family, the optimal running time $f(k)n^{g(k)}$ required to find satisfying assignments that set precisely $k$ of the $n$ variables to $1$. 

Under central hardness assumptions on detecting cliques in graphs and 3-uniform hypergraphs, we give an almost tight characterization of $g(k)$ into four regimes: 
\begin{enumerate}
\item Brute force is essentially best-possible, i.e., $g(k) = (1\pm o(1))k$,
\item the best algorithms are as fast as current $k$-clique algorithms, i.e., $g(k)=(\omega/3\pm o(1))k$,
\item the exponent has sublinear dependence on $k$ with $g(k) \in [\Omega(\sqrt[3]{k}), O(\sqrt{k})]$, or
\item the problem is fixed-parameter tractable, i.e., $g(k) = O(1)$.
\end{enumerate}

This yields a more fine-grained perspective than a previous $\FPT$/$\W[1]$-hardness dichotomy (Marx, Computational Complexity 2005). Our most interesting technical contribution is a $f(k)n^{4\sqrt{k}}$-time algorithm for \SubsetSum with precedence constraints parameterized by the target $k$ -- particularly the approach, based on generalizing a bound on the Frobenius coin problem to a setting with precedence constraints, might be of independent interest.   
\end{abstract}

\section{Introduction}

Extensive research in complexity theory has established methods to give precise qualitative results on the computational hardness of problems. In this context, a basic question that we would like to answer is: When are there algorithms better than a brute force search, and if there are, how much improvement is possible compared to brute force? In problem settings where the task is to find a solution of size $k$, typically it is easy to obtain algorithms with running time of the form $\Oh(n^{k+\Oh(1)})$ by a brute force search of every possible solution. In such cases, beating brute force could involve having an algorithm with a term $(1-\epsilon)k+O(1)$ in the exponent for some $\epsilon>0$, or having sublinear (e.g, $O(k/\log k)$ or $O(\sqrt{k})$) dependence on $k$ in the exponent, or we might be able to completely remove $k$ from the exponent of $n$ with an $f(k)n^{O(1)}$ time algorithm.

In this paper, we study the above question in the context of the class of Boolean Constraint Satisfaction problems. Fixing a \emph{constraint family} $\confam$ of Boolean functions, the task is to determine an assignment to Boolean variables $x_1,\dots, x_n$ satisfying a given conjunction of constraints of the form $f(x_{i_1},\dots, x_{i_r})$ with $f\in \confam$ and $i_1,\dots, i_r \in [n]$. Here, the natural notion of the solution size $k$ is the number of variables set to $1$ and we consider the task of determining a satisfying assignment with precisely $k$ ones. 
This class indeed contains a variety of problems: basic graph problems such as the vertex cover problem ($\confam$ consists of the binary OR) and the independent set problem in graphs ($\confam$ consists of the binary NAND) or $d$-uniform hypergraphs ($\confam$ consists of the $d$-ary NAND), but also other natural problems such as a formulation of \SubsetSum parameterized by the target $k$ ($\confam$ consists of binary equality)\footnote{To see the correspondence, note that if $\confam$ consists of the binary equality, $\SAT(\confam)$ asks to find a union of connected components of total size $k$. By representing each connected component by its size (after linear-time preprocessing), this is precisely the \SubsetSum problem with target $k$.}, finding a solution of a (sparse) linear system over GF(2) where each linear equality involves at most a constant number $r$ of variables and the solution must have precisely $k$ ones ($\confam$ consists of all linear constraints of arity at most $r$), as well as finding a closed set of size~$k$ in a directed graph ($\confam$ consists of the binary implication). Note that the last problem can be seen to be equivalent to a variant of \SubsetSum that prescribes precedence constraints on the items and uses an unary encoding for all item sizes.

The time complexity inside this class varies widely: Vertex cover is famously fixed-parameter tractable when parameterized by $k$, with a best current running time bound of $\Oh(kn+2^{O(k)})$~\cite{ChenKX10}. It is even simpler to solve the \SubsetSum formulation in time $\Oh(m+k^2)=\Oh(n^2)$ (where $m$ is the number of edges in the graph) by a straightforward algorithm\footnote{Determine all connected components in time $\Oh(m)$ and solve a \SubsetSum instance on the component sizes in time $\Oh(k^2)$ using Bellman's pesudopolynomial-time algorithm or recent improvements~\cite{KoiliarisX19, Bringmann17}.}. The fastest known algorithm for independent set~\cite{NesetrilP85}, however, relies on the sophisticated techniques for matrix multiplication, and achieves a running time of $\Oh(n^{(\omega/3)k})$ for $k$ divisible by 3, where $\omega \le 2.373$ is the matrix multiplication exponent. For finding closed sets of size $k$, a surprisingly simple $\Oh(n^{k/2})$-time algorithm\footnote{Without loss of generality, it suffices to solve the following problem: given a node-weighted DAG $G=(V,E)$ and $k\in \nats$, find a weight-$k$ subset $S\subseteq V$ such that $u \in S$ and $(u,v)\in E$ implies $v\in S$. If $S$ contains a set $S'$ of at most $k/2$ \emph{sources} (i.e., vertices that have no incoming edges from other vertices in $S$), we can simply guess $S'$ and check that $S'$ and the set of all descendants of $S'$ have total weight $k$. If $S$ contains no such set $S'$ of size at most $k/2$, we can guess all $\le k/2$ non-sources $S''$, remove all incoming edges to $S''$ and find a weight-$(k-|S''|)$ set of vertices with out-degree $0$.} improves over brute force even without matrix multiplication, but a priori there is little indication for the optimality of this approach. Finally, for finding independent sets in 3-uniform hypergraphs, no substantially faster-than-brute-force algorithm is known. 

The central purpose of this paper is to give a detailed understanding of the time complexity of Boolean constraint satisfaction parameterized by solution size $k$, particularly when $k$ is considered a (large) constant: How precisely can we determine the running time $f(k) n^{g(k)}$, with $g(k)$ as small as possible? Note that for large constant $k$, we have $f(k) n^{g(k)}= \Oh(n^{g(k)})$ and aim to determine its optimal polynomial-time complexity. 

A classification of the second author~\cite{Marx05} resolves the qualitative question for which $\confam$ the problem is solvable in FPT time (assuming $\FPT \ne \W[1]$), i.e., when $g(k)$ can be bounded by a constant independent of $k$. In particular, from this classification, we obtain that among the above examples, vertex cover, \SubsetSum with target $k$, and the sparse linear systems over GF(2) can be solved in time $f(k) n^{c}$, while for independent set (in both graphs and hypergraphs) as well as \SubsetSum with precedence constraints, the exponent of $n$ must depend on $k$ (unless $\FPT = \W[1]$). Can we obtain tight bounds on $g(k)$ when it must depend on $k$? In particular, can we determine for which $\confam$ the brute-force $\Oh(n^{k+c})$-time solution is essentially optimal? %

\subsection{Our Results}

Let us formally state our problems and results.

\begin{problem}
Let $\confam$ be a finite \emph{constraint family} of Boolean functions. The problem $\SAT(\confam)$ asks to determine whether a given formula $\phi$ on Boolean variables $x_1,\dots,x_n$ is satisfiable by an assignment with $k$ ones, where $\phi$ is a conjunction of $m$ constraints $C$ of the form $f(\mathbf{x})$, where $f:\{0,1\}^r \to \{0,1\}$ is a constraint function in $\confam$ and $\mathbf{x}$ is an $r$-tuple of variables among $x_1,\dots, x_n$.
\end{problem}

Note that if all $f\in \confam$ have arity bounded by $r$, then there are at most $\Oh(n^r)$ possible constraints, and exhaustive search solves $\SAT(\confam)$ in time $\Oh(n^{k+r})$.

We will show that the complexity of $\SAT(\confam)$ is tightly characterized by the set of functions \emph{expressible as restrictions of constraint functions $f\in \confam$}. To formally introduce this concept, let $f: \{0,1\}^r \to \{0,1\}$ be an arbitrary Boolean function. We say that $g:\{0,1\}^s\to \{0,1\}$ is a \emph{restriction of $f$} if it is obtained from $g$ by replacing each argument of $f$ by either the constant $0$, the constant $1$, or an argument of $g$, i.e., we can partition $[r]$ into $X_1,\dots, X_s, Z_0, Z_1$ such that 
\[ g(x_1,\dots,x_s) = f(\overbrace{\entry{x_1}}^{X_1}, \dots, \overbrace{\entry{x_s}}^{X_s}, \overbrace{\entry{0}}^{Z_0}, \overbrace{\entry{1}}^{Z_1}).\]
Here, $\overbrace{\entry{y}}^Y$ denotes plugging in $y$ for all (not necessarily contiguous) positions $Y\subseteq [r]$, see Section~\ref{sec:prelim}. 

\begin{defn}
Let $g:\{0,1\}^d\to \{0,1\}$ be an arbitrary Boolean function.
A constraint family $\confam$ \emph{represents} $g$ if there is some $f\in \confam$ such that $g$ is a restriction of $f$. If $\confam$ does not represent~$g$, we say that $\confam$ \emph{avoids} $g$.
\end{defn}

Let $\Impl:\{0,1\}^2 \to \{0,1\}$ and $\NAND_d :\{0,1\}^d \to \{0,1\}$  be the binary implication and $d$-ary $\NAND$ function, respectively, i.e., 
\begin{align*}
\Impl(y_1,y_2) & \coloneqq \overline{y_1} \vee y_2,\\
\NAND_d(y_1,\dots,y_d) & \coloneqq \overline{\bigwedge\nolimits_{i=1}^d y_i}.  
\end{align*}

In~\cite{Marx05}, it is shown that $\SAT(\confam)$ is solvable in FPT time $f(k)n^{c}$ if and only if $\confam$ is \emph{weakly separable}, which is a condition equivalent to $\confam$ avoiding $\NAND_2$ and $\Impl$.
We show an almost tight characterization of $g(k)$ (under plausible assumptions from fine-grained complexity theory) that depends only on whether or not $\confam$ represents $\Impl$, $\NAND_2$ or $\NAND_d$ for higher order $d\ge 3$. Specifically, we obtain the following main theorem, illustrated in Figure~\ref{fig:results}.

\begin{figure}
\centering
\includegraphics[width=0.72\textwidth]{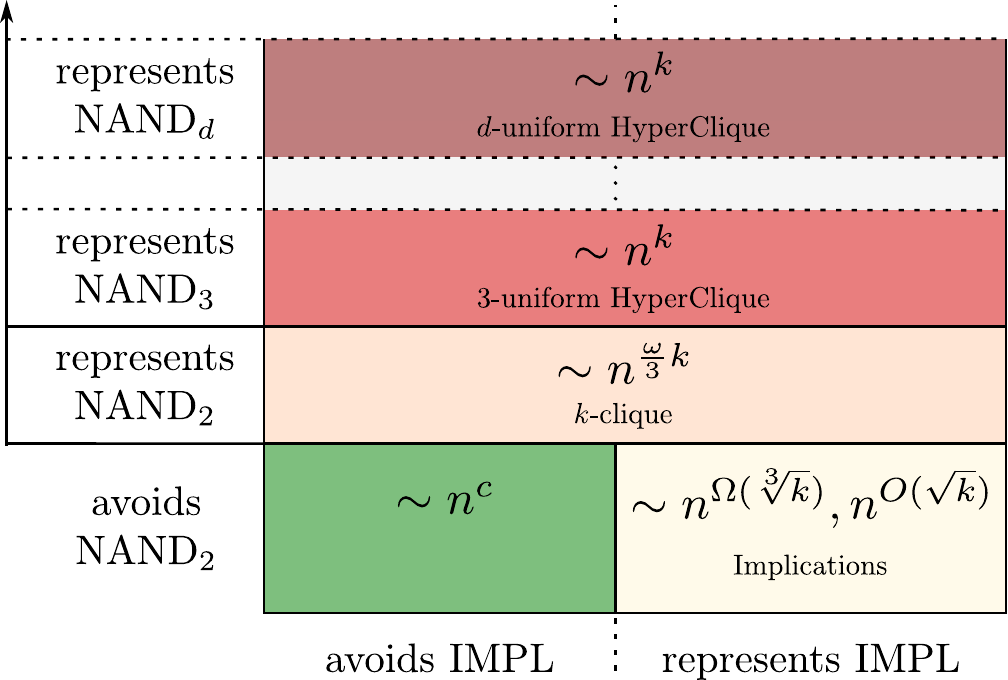}
\caption{Overview over our main results. The parts of the diagram to the right of the vertical $\Impl$ line depict $\confam$ representing $\Impl$, while the parts to the left avoid $\Impl$. Analogously, the parts of the diagram above a $\NAND_d$ line depict $\NAND_d$-representing $\confam$, while those below avoid $\NAND_d$. For each cell, we illustrate our (typically matching) algorithmic and hardness results, together with a problem that is complete for this cell (in a certain sense). For clarity of presentation, we drop additional $f(k)n^c$-factors of stated running times.}
\label{fig:results}
\end{figure}

\newcommand{\regime}[1]{{\textbf{\textup{[#1]}}}}
\begin{thm}
Let $\confam$ be a finite constraint family.
\begin{enumerate}
\item \regime{FPT regime} If $\confam$ avoids both $\NAND_2$ and $\Impl$, then there is a computable $f(k)$ and constant~$c_\confam$ such that $\SAT(\confam)$ can be solved in time $f(k)n^{c_\confam}$.
\item \regime{Subexponential regime}\newline If $\confam$ represents $\Impl$, but avoids $\NAND_2$, then there is a computable $f(k)$ and constant~$c_\confam$ such that $\SAT(\confam)$ can be solved in time $f(k)n^{4\sqrt{k}+c_\confam}$; \newline furthermore, for no computable $f(k)$ and constants $c_\confam, \eps > 0$, $\SAT(\confam)$ can be solved in time  $f(k)n^{(\omega/6-\eps)\sqrt[3]{k}+c_\confam}$, unless the $k$-clique conjecture fails.
\item \regime{Clique regime}\newline  If $\confam$ represents $\NAND_2$, but avoids $\NAND_3$, then there is a computable $f(k)$ and constant~$c_\confam$ such that $\SAT(\confam)$ can be solved in time $f(k) n^{(\omega/3)k+c_\confam}$; \newline furthermore for no computable $f(k$) and constants $c_\confam, \eps > 0$, $\SAT(\confam)$ can be solved in time $f(k) n^{(\omega/3-\eps)k+c_\confam}$, unless the $k$-clique conjecture fails. 
\item \regime{Brute-force regime}\newline  If $\confam$ represents $\NAND_3$, then for no computable $f(k)$ and constants $c_\confam, \eps > 0$, $\SAT(\confam)$ can be solved in time $f(k)n^{(1-\eps)k+c_\confam}$, unless the 3-uniform $k$-HyperClique conjecture fails. 
\end{enumerate}
\end{thm}

That is, we only have four regimes: $g(k)$ is either constant, sublinear in $k$ with a value between essentially $(\omega/6) \sqrt[3]{k}$ and $4\sqrt{k}$, the clique detection bound of essentially $(\omega/3) k$, or the brute force bound of essentially $k$.  Note that we do not try to optimize the bounds on $f(k)$, which generally are bounded by $r^{\Oh(k^3)}$, where $r$ is the arity of $\confam$.

Let us briefly discuss our hardness assumptions and their plausibility (for a detailed discussion, we refer to Section~\ref{sec:hardnessassumptions}): The $k$-clique conjecture postulates that there is no $\Oh(n^{(\omega/3-\epsilon)k+c})$ time algorithm for detecting a $k$-clique in a given graph, with a matching upper bound of $\Oh(n^{(\omega/3)k+1})$ known since 1985~\cite{NesetrilP85}. By now, it has been used, e.g., to justify (conditional) optimality of Valiant's parser for context free grammars~\cite{AbboudBVW18} and to give conditional lower bounds for string problems~\cite{BringmannGL17, AbboudBBK17}, average-case settings~\cite{Boix-AdseraBB19}, and more. Notably, the only $k$-clique algorithm known to break brute force by a polynomial factor makes crucial use of fast matrix multiplication techniques -- unfortunately, these techniques do not extend to finding cliques in \emph{hypergraphs}. This has led to the $d$-uniform HyperClique conjecture (for arbitrary $d\ge 3$): This conjecture states that there is no algorithm beating brute force, i.e., no $\Oh(n^{(1-\eps)k+c})$-time algorithm, for detecting a $k$-clique in a given $d$-uniform hypergraph. It has been used to expose hardness of problems in sparse graphs~\cite{LincolnVWW18}, for first-order queries to relational databases (specifically, in model-checking~\cite{BringmannFK19} and enumeration contexts~\cite{CarmeliK19}), and for the orthogonal vectors problem~\cite{AbboudBDN18}; furthermore, it is known that its refutation requires giving a $\Oh((2-\eps)^n)$-time algorithm for Max-3SAT -- we refer to~\cite{AbboudBVW18,LincolnVWW18} for more detailed discussions of the plausibility of the ($d$-uniform Hyper-)Clique conjecture.

Interestingly, our classification does not fundamentally rely on the validity of the $d$-uniform HyperClique conjecture: If, for some $d\ge 3$, the $d$-uniform HyperClique conjecture is eventually refuted, we obtain faster-than-brute-force algorithms for all $\NAND_{d+1}$-avoiding families!  

\paragraph{Coarser Classification.}
While we state our results under very fine-grained hardness assumptions on clique and hyperclique detection, we may also state a coarser classification assuming only the assumption that $k$-clique cannot be solved in time $f(k)n^{o(k)}$. Already under this assumption, which is implied by the Exponential Time Hypothesis (see~\cite{ChenCFHJKX05, ChenHKX06}), our reductions and algorithms show that there exists an FPT regime where $g(k)$ is a constant, a subexponential regime where $g(k)$ is between $\Omega(\sqrt[3]{k})$ and $\Oh(\sqrt{k})$, and a linear regime where $g(k) = \Theta(k)$. However, based on the Exponential Time Hypothesis only, we cannot distinguish problems solvable in time $f(k)n^{(1\pm o(1))k}$ and $f(k)n^{(\nicefrac{\omega}{3}\pm o(1)) k}$, and thus cannot differentiate in the linear regime.

\paragraph{Examples.}
From our general classification, we can draw some interesting specific corollaries (assume here that $k$ is a large constant):

\textbf{3-SAT:} Finding satisfying assignments with $k$ ones for $3$-CNF formulas ($\confam$ consists of all ternary functions with a single falsifying assignment) requires brute force time $n^{(1-o(1))k}$ under the 3-uniform HyperClique conjecture. However, if we drop a single function from $\confam$ (specifically $\NAND_3$, i.e., each constraint must have at most two negative literals), the problem can be solved in time $\Oh(n^{(\omega/3)k+c})$, which is essentially optimal under the $k$-Clique conjecture.

\textbf{Subexponential cases:} We obtain $n^{\Oh(\sqrt{k})}$-time algorithms for interesting special cases: Beyond precedence-constrained \SubsetSum with target $k$ (i.e, $\SAT(\{\Impl\})$), this includes $\SAT(\{f\})$ with $f(y_1,y_2,y_3) \coloneqq  y_1 \Rightarrow (y_2 \vee y_3)$, and, more generally, every finite set of \emph{dual-Horn constraints} (i.e., constraints that can be represented by clauses with at most a single negative literal)\footnote{It is known that a constraint is dual-Horn if and only if it its satisfying assignments are closed under union, which immediately implies that it cannot contain $\NAND_2$ as a restriction.}. This also includes examples beyond dual-Horn constraints such as $\SAT(\{\Impl, f'\})$ with $f'$ being defined by $f'(y_1, y_2,y_3) = 1$ iff $(y_1,y_2,y_3) \in \{(0,0,0), (1,0,1), (1,1,0)\}$. Interestingly, all of these problems have the same (conditionally optimal) time complexity of $f(k) n^{\Theta(k^\alpha)}$ with $1/3 \le \alpha \le 1/2$; determining the precise value of $\alpha$ remains a challenge for future work. 

\subsection{Technical Overview}

We give an overview of the technical challenges that are handled in our work, from the highest running time regime to the lowest running time regime:

\textbf{Brute-force regime:}
It is straightforward to obtain hardness for $\NAND_3$-representing families by the following intuitive approach: To reduce from $k$-clique in a 3-uniform hypergraph~$G$, we let $x_i$ denote whether we include vertex $v_i$ in our $k$-clique. By the standard observation that a clique in a hypergraph $G$ is an independent set of its complement graph $\overline{G}$, we only need to ensure that for each edge $e=(v_a,v_b,v_c)$ of $\overline{G}$, not all vertices are included in our clique, i.e., $\NAND_3(x_a,x_b,x_c)$ holds. Since $\confam$ represents $\NAND_3$, we can express this constraint using an appropriate restriction of some $f\in \confam$. Here, there is a technical issue of how we can generate the constants $0$ or $1$ to obtain the desired restrictions -- using not particularly difficult, but careful constructions, we show that we can always simulate these constants as needed (Section~\ref{sec:hardness}).

\textbf{Moderately hard regime:}
While the hardness of $\NAND_3$-representing families is straightforward, it is surprising that this condition is in fact necessary for the brute-force approach to be (conditionally) optimal: If $\NAND_3$ is \emph{not} representable, we give a $f(k)n^{(\omega/3)k +c_\confam}$-time algorithm via reduction to $k$-Clique. 

The essential idea for this reduction is the following win-win argument. Let us denote by $a_{x,y}$ the weight-2 assignment setting only $x$ and $y$ to 1. Fix any weight-$k$ satisfying assignment~$a$. If there are two variables $x_i=x_{i'}=1$ in $a$ such that $a_{x_i,x_{i'}}$ is not satisfying, then we can use this pair of variables to ``guide'' our search towards $a$. We guess $x_i,x_{i'}$, identify a falsified constraint (of arity $r$) and guess an additional third variable from the at most $r-2$ other variables in this constraint. This means that by guessing two variables ($n^2$ possibilities), we obtain an additional variable almost for free (guessing $r-2$ possibilities). That is, in the considered case we can identify 3 variables of $a$ with a guess of $(r-2)n^2$ possibilities, which is a significant gain compared to the $n^3$ possibilities of brute force. Otherwise, if $a$ has no such pair of variables, we observe that $a$ satisfies already a simpler formula that uses only $\NAND_2$'s: specifically, the conjunction of $\NAND(x_i,x_{i'})$ for all $i,i'$ such that assignment $a_{x_i,x_{i'}}$ violates the original formula. Furthermore, we show that since $\NAND_3$ is not representable, any solution of the simpler formula indeed remains a solution of the original formula.

Interestingly, this reduction generalizes also to hypergraphs so that a refutation of the $d$-uniform HyperClique conjecture would give a $f(k)n^{(1-\eps)k+c_\confam}$-time algorithm for $\NAND_{d+1}$-avoiding families.  

On the hardness side, analogously to the brute-force regime, it is rather straightforward to show that $k$-clique running time is indeed necessary for $\NAND_2$-representing constraint families (see Section~\ref{sec:hardness}), which thus concludes a tight bound on $g(k)$ of essentially $(\omega/3)k$ in this regime.

\textbf{Mildly hard regime:} This is the technically most interesting regime. If $\NAND_2$ is not representable, then $\SAT(\confam)$ might still not have an FPT algorithm, specifically, if it represents $\Impl$. Implicit in the $W[1]$-hardness proof in~\cite{Marx05} is a fine-grained lower bound of $n^{\Omega(\log k)}$ under the $k$-clique conjecture. By giving a careful adaptation of the lower bound of~\cite{Marx05}, we can strengthen this lower bound to $n^{\Omega(\sqrt[3]{k})}$. While it is conceivable that this lower bound can be strengthened to $n^{\Omega(\sqrt{k})}$, the structure of the construction suffers from a fundamental obstacle that makes a lower bound beyond $n^{\Omega(\sqrt{k})}$ seem unlikely. This raises the suspicion that a $n^{o(k)}$-time algorithm for $\NAND_2$-avoiding families could exist -- and indeed, we manage to develop a $n^{\Oh(\sqrt{k})}$-time algorithm, which is perhaps the most interesting technical contribution of our paper.   

To illustrate our approach, consider the problem \wDAGImplication: Given a DAG $G=(V, E)$ with node weights $w:V \to \nats$ and a parameter $k\in \nats$, the task is to find a set $S\subseteq V$ such that (1) $u\in S$ and $(u,v)\in E$ implies $v\in S$ and (2) $S$ has total weight $\sum_{s\in S} w(s) = k$. Without edges, this problem simplifies to \SubsetSum which we could solve in $\poly(k)$ time~\cite{KoiliarisX19, Bringmann17}. However, to enable a generalization to our precedence setting, we describe a different approach based on a combinatorial property inspired by the famous Frobenius coin problem: Given coins of denominations $2 \le d_1< d_2 < \dots < d_\ell$ with $\gcd(d_1,\dots, d_\ell)=1$, what is the largest number $x$ not representable as $x=\sum_{i=1}^\ell \alpha_i d_i$ for some non-negative values $\alpha_i \ge 0$? A proof attributed to Schur (see \cite{Bauer42, RamirezAlfonsin05, Gallier14}) yields an upper bound of $x\le (d_1-1)(d_\ell-1)$. Consequently, if $w_1\le \dots\le  w_\ell$ with $\gcd(w_1,\dots,w_\ell)\divides k$ are the weights occurring in an edgeless $G$, and $w_\ell\le \sqrt{k}$, then there always exists a set $S$ of total weight $k$, provided each weight occurs sufficiently often (say, at least $k$ times). Thus, if we can preprocess the instance such that each weight is bounded by $\sqrt{k}$ and occurs sufficiently often, we can determine the answer to the instance by simply computing the gcd of the weights. Intuitively, this is possible in time $n^{\Oh(\sqrt{k})}$ by guessing the $\Oh(\sqrt{k})$ vertices of weight larger than $\sqrt{k}$, as well as brute-forcing vertices of each weight class containing only few vertices.

Interestingly, this approach can be lifted to the setting with precedence constraints. To this end, assume that the graph consists of layers $V_1,\dots, V_\ell$ such that each $V_i$ consists of a sufficiently large number of vertices of weight $w_i$ and that all edges respect the layering (i.e., an edge between a vertex in $V_i$ and a vertex in $V_j$ implies $i> j$). We show the following property, which gives a generalization of Schur's bound to the precedence setting: \begin{center}If for each vertex $v$, the total weight of all its descendants (including $v$ itself) is at most $\sqrt{k/2}$, then there exists a solution of total size $k$ \emph{if and only if} $\gcd(w_1,\dots,w_\ell) \divides k$.\end{center} 
By an $n^{\Oh(\sqrt{k})}$-time preprocessing analogous to the intuitive arguments for the edge-less case, we can ensure that the preconditions are satisfied. 
We give the details of this approach in Section~\ref{sec:implalgo}.

The above algorithmic insight solves the \wDAGImplication problem in time $\Oh(n^{4\sqrt{k}})$. To obtain such a bound for all $\NAND_2$-avoiding families, we use a randomized reduction to \wDAGImplication. On a very high level, the approach is to create a \wDAGImplication instance $G$ that contains only solutions that satisfy the given formula $\phi$ by iteratively choosing random implications consistent with certain solutions of $\phi$. Doing this in an appropriate manner, a fixed feasible solution survives this process with $1/f(k)$ probability, which gives an algorithm running in time essentially $\Oh(f(k) n^{4\sqrt{k}})$. We give the details in Section~\ref{sec:redtoImpl}.

\textbf{Fast regime:}
For the remaining regime of families avoiding both $\Impl$ and $\NAND_2$, an $f(k)n^c$-time algorithm follows from~\cite{Marx05}, concluding the characterization.

\subsection{Related work}

Dichotomy theorems for constraint satisfaction have a rich history, starting with Schaefer's Theorem classifying Boolean Constraint Satisfaction Problems (CSPs) into either polynomial-time solvable or NP-complete~\cite{Schaefer78}. The subsequent \emph{Dichotomy conjecture}~\cite{FederV98}, which postulated that Schaefer's Theorem can be extended to any constant domain size beyond Boolean, was resolved positively only recently by Bulatov~\cite{Bulatov17} and Zhuk~\cite{Zhuk17}.    
Further classifications have been investigated in a number of related settings, including quantified CSP (see, e.g.,~\cite{CreignouKS01, ZhukM20}) and optimization variants (see, e.g.~\cite{Creignou95, KhannaSTW00}). Parameterizing by the solution size (as we do here), corresponding dichotomies have been obtained for Boolean~\cite{Marx05} and larger domain sizes~\cite{BulatovM14, Lin18}, with a characterization of kernelization for Boolean domain given in~\cite{KratschMW16} and a study of parameterized approximability given in~\cite{BonnetEM16}. A parameterized dichotomy for related local search tasks has been given in~\cite{KrokhinM12}. 

On a conceptual level, our work is related to a fine-grained classification result for model-checking first-order properties with a bounded number of quantifiers~\cite{BringmannFK19}, where a fine-grained dichotomy under the 3-uniform HyperClique conjecture is given. Note, however, that the hardness criterion and techniques developed there are substantially different due to the different nature of the problem settings.   
 
\subsection{Open Problems}

The main open problem raised by our work is to close the gap in the subexponential regime: Can we solve $\Implication =\SAT(\{\Impl\})$ already in $f(k)n^{\Oh(\sqrt[3]{k})}$ or can we improve our lower bound to $f(k)n^{\Omega(\sqrt{k})}$? 
Note that by our reductions, improved bounds directly transfer to all $\NAND_2$-avoiding families. 

Second, a natural direction is to extend our classification beyond the Boolean domain, i.e., give a fine-grained perspective building on~\cite{BulatovM14,Lin18}. 

Finally, interesting related settings include natural problem variants with different size restrictions (\emph{at most $k$} or \emph{at least $k$}), local search tasks as well as optimization settings with weights on the variables or on the constraints.  

\section{Preliminaries}
\label{sec:prelim}

We write $[n]\coloneqq \{1,\dots, n\}$ and for any set $S$ and integer $d$, let ${S \choose d}$ denote the set of $d$-element subsets of $S$.

For a finite constraint family $\confam$, we say its \emph{arity $r$} is the maximum arity of a function $f\in \confam$. Since in the constraints of $\SAT(\confam)$, we may use variables in arbitrary order, we use the following notation for convenience: For any $f:\{0,1\}^r \to \{0,1\}$ and partition $X_1,\dots, X_s$ of $[r]$, we write  
\[ f(\overbrace{\entry{x_1}}^{X_1},\dots, \overbrace{\entry{x_s}}^{X_s})\]
to denote the value of $f(u_1,\dots,u_r)$ where we plug in $x_j$ for each $u_i$ with $i\in X_j$. Correspondingly $g:\{0,1\}^d \to \{0,1\}$ can be obtained as a restriction of $f$ if and only if there is a partition $X_1,\dots, X_d, Z_0, Z_1$ of $[r]$ such that 
\[ g(x_1,\dots,x_d) = f(\overbrace{\entry{x_1}}^{X_1},\dots, \overbrace{\entry{x_d}}^{X_d}, \overbrace{\entry{0}}^{Z_0}, \overbrace{\entry{1}}^{Z_1}).\]

We say that an assignment $a: [n] \to \{0,1\}$ has \emph{weight} $k$ if $\sum_{i=1}^n a(i) = k$. Furthermore, we say that $a$ is dominated by an assignment $a': [n]\to \{0,1\}$, written $a\le a'$, if for all $i\in [n]$, we have $a(i) \le a'(i)$. For a subset $S\subseteq [n]$, we let $a_S$ denote the assignment that sets $a(i)=1$ if and only if $i\in S$. We let $\ones(a) \coloneqq \{ x_i \mid a(i) = 1\}$ denote the set of $1$-variables of $a$. For any constraint $C=f(\mathbf{x})$ where $\mathbf{x} = (x_{i_1},\dots,x_{i_r})$ with $i_1,\dots, i_r\in [n]$, we let $\vars(C) = \{x_{i_1},\dots,x_{i_r}\}$ denote the variable set involved in $C$.

All graphs considered in this paper are simple, i.e., we disallow multiple edges and self-loops. If $G=(V,E)$ is a directed graph, we call $S\subseteq V$ a \emph{closed} set if for all $(u,v)\in E$, we have that $u\in S$ implies that $v\in S$. We say that $v$ is a \emph{descendant} of $u$ if $v$ is reachable by a path from $u$ and let $D(u)$ denote the set of descendants of $u$ (including $u$ itself). Analogously, if $v$ is a descendant of $u$, we call $u$ an \emph{ascendant} of $v$. We extend the notation naturally to sets $S\subseteq V$ by defining $D(S) \coloneqq \bigcup_{u\in S} D(u)$.
For a graph $G=(V,E)$ with node weights $w: V\to \nats$ and $S\subseteq V$, we write $w(S) \coloneqq \sum_{v\in S}w(v)$. For any $S\subseteq V$, we let $G[S]$ denote the subgraph of $G$ induced by $S$, i.e., the subgraph obtained by deleting all vertices in $V\setminus S$ and adjacent edges.

\subsection{Hardness Assumptions}

Let $k$-clique denote the following problem: Given a (simple) undirected graph $G=(V,E)$, determine whether there is a \emph{clique} of size $k$, i.e., $S\subseteq V, |S|=k$ such that for all $\{u,v\} \in {S \choose 2}$ we have $\{u,v\} \in E$. A simple algorithm~\cite{NesetrilP85} solves $k$-clique in time $\Oh(n^{\omega/3 k})$ when $k$ is divisible by~$3$, which extends to time $\Oh(n^{\lfloor k/3\rfloor \omega + (k\bmod 3)})$ for arbitrary $k$ (for more precise bounds, see~\cite{EisenbrandG04}).
\label{sec:hardnessassumptions}
This running time is conjectured to be best possible, in the following sense.

\begin{hypo}[$k$-Clique Conjecture]
For no $c,\eps > 0$ and $f(k)$, there is an $f(k)n^{(\omega/3 - \epsilon)k+c}$-time algorithm for $k$-Clique.\footnote{Note: sometimes, the $k$-clique conjecture is stated as \[\inf \{ F \mid \text{$3k$-clique can be solved in time $n^{Fk+o(1)}$ for all (sufficiently large) constant $k$}\} = \omega,\] which can be seen to be equivalent to the above formulation via a standard self-reduction for $k$-clique.}
\end{hypo}
As without the use of matrix multiplication, no $\Oh(n^{(1-\eps)k+c})$-time algorithms are known, a variant of the conjecture postulates that there are even no $\Oh(n^{(1-\eps)k+c})$-time \emph{combinatorial} algorithms, i.e., algorithms avoiding the sophisticated algebraic techniques underlying current matrix multiplication algorithms. 
  
  By now, the $k$-clique conjecture has been used to explain hardness barriers in various contexts, such as the optimality of Valiant's parser for context-free grammar recognition~\cite{AbboudBVW18},  pattern matching in uncompressed and compressed strings~\cite{BringmannGL17, AbboudBBK17}, average-case hardness~\cite{Boix-AdseraBB19} and more. %
  For a more detailed discussion of this hardness assumption, we refer to~\cite{AbboudBVW18}.

  The $k$-clique problem naturally extends to hypergraphs: Given a $d$-uniform hypergraph $G=(V,E)$, the $d$-uniform $k$-HyperClique problem asks to determine whether there is a \emph{(hyper-)clique} of size $k$, i.e., $S\subseteq V, |S|=k$ such that for all subsets $S' \in {S \choose d}$, we have $S'\in E$. 

\begin{hypo}[$d$-Uniform $k$-HyperClique Conjecture]
Let $d\ge 3$. For no $c,\eps > 0$ and $f(k)$, there is an $f(k)n^{(1 - \epsilon)k+c}$-time algorithm for $d$-uniform $k$-HyperClique.
\end{hypo}

Similarly to the $k$-Clique conjecture, this hardness conjecture reveals hardness barriers in a number of contexts, such as hardness for problems on sparse graphs~\cite{LincolnVWW18}, for deciding or enumerating answers to first-order queries~\cite{BringmannFK19, CarmeliK19} and for the study of fine-grained average-case complexity~\cite{Boix-AdseraBB19}. It is known that it implies the Orthogonal Vectors conjecture~\cite{AbboudBDN18}, however, refuting this conjecture requires (at least) to give an $\Oh((2-\eps)^n)$-time exact algorithm for Max3SAT; for details and further discussion of the plausibility of this conjecture, we refer to \cite{LincolnVWW18}.   

\section{Algorithm for \Implication}
\label{sec:implalgo}

In this section, we give an algorithm for the problem $\Implication = \SAT(\{\Impl\})$ that is much faster than brute force and achieves $O(\sqrt{k})$ dependence of $k$ in the exponent $n$. 
For convenience, we reduce \Implication to the following problem. (Recall that for any graph $G=(V,E)$, we say that $S\subseteq V$ is \emph{closed}, if for all $(u,v)\in E$, we have $u \in S$ implies $v\in S$.)

\begin{problem}[\wDAGImplication]
Given an DAG $G=(V,E)$ with node weights $w: V\to \nats$ and parameter $k\in \nats$, determine whether there is a closed set $S \subseteq V$ of weight exactly $k$, i.e., $w(S) = k$.
\end{problem}
The easy reduction works as follows. For each variable $x_i$, we introduce a corresponding vertex $x_i$ of weight 1 and introduce an edge $(x_i, x_j)$ for every implication constraint $x_i \Rightarrow x_j$ of $\phi$. We contract each strongly connected component $C=\{v_1, v_2, \dots, v_\ell\}$ in $G$ to a single vertex $v_C$ of weight $\sum_{i=1}^\ell w(v_i)$ in time $\Oh(n+m)= \Oh(n^2)$~\cite{Tarjan72}. Observe that the resulting graph is a DAG which has a closed set of weight~$k$ if and only if $\phi$ has satisfying assignment of weight $k$.

Recall that for any $v\in V$, we let $D(v)$ denote the set of descendants of $v$, i.e., the set of nodes reachable from $v$ (including $v$).

As we will formally argue later, by a $f(k)n^{\Oh(\sqrt{k})}$-time preprocessing it is not difficult to preprocess a \wDAGImplication instance into the following form, which we call \emph{Frobenius instance}, as it admits a combinatorial characterization of solvability that is analogous to Schur's bound for the Frobenius coin problem.

\begin{defn}
A \emph{Frobenius instance with parameter $k$} is a weighted directed graph $G = (V,E, w)$ with $\ell$ parts $V = V_1 \cup V_2 \cup \cdots \cup V_\ell$ and weight function $w:V \to \nats$ such that the following properties hold: 
\begin{enumerate}[label=(P\arabic*)]
\item\label{item:structweights} there are weights $w_1, \dots, w_\ell$ such that $w(v) = w_i$ for all $v\in V_i$ and $i \in [\ell]$.
\item\label{item:structedges} for any edge $(u,v)\in E$, we have $u\in V_i$ and $v\in V_j$ for some $\ell \ge i > j \ge 1$,
\item\label{item:size} for all $i\in [\ell]$, we have $|V_i| \ge k$, 
\item\label{item:weights} for all $v\in V$, we have $w(D(v)) \le \sqrt{k/2}$.
\end{enumerate}
\end{defn}

Intuitively, the necessary preprocessing follows from the following arguments: To ensure~\ref{item:weights}, note that any weight-$k$ closed set $S$ has at most $\sqrt{2k}$ many vertices $v\in S$ with $w(D(v)) > \sqrt{k/2}$, which we can exhaustively enumerate with $n^{\Oh(\sqrt{k})}$-time overhead. By suitably arranging remaining nodes among the layers, it is straightforward to ensure~\ref{item:structweights}, \ref{item:structedges} and additionally that $\ell \le f(k)$, since by~\ref{item:weights}, each node has at most $\Oh(\sqrt{k})$ descendants. Finally, to ensure~\ref{item:size}, if any part $V_i$ is small (i.e., $|V_i| < k$), we can exhaustively try out including any subset of $V_i$, introducing an overhead of only $2^{\Oh(k)}$ per $V_i$; since $\ell \le f(k)$, this additional overhead is bounded by $f(k)2^{\Oh(k)}$.

If a Frobenius instance had no edges, then Schur's bound on the Frobenius coin problem implies that it has a solution if and only if $\gcd(w_1,\dots, w_\ell) \divides k$. We prove that this criterion holds even in the setting of precedence constraints.  

\begin{lem}\label{lem:Frobcrit}
Let $G$ be a Frobenius instance with  parameter $k$. Then $G$ has a closed set of weight~$k$ if and only if $\gcd(w_1,\dots, w_\ell) \divides k$.
\end{lem}
\begin{proof}
Since $\gcd(w_1,\dots, w_\ell) \divides w(S)$ for any $S\subseteq V$, the condition $\gcd(w_1,\dots,w_\ell) \divides k$ is necessary for $G$ to have a closed set of weight $k$.

We show that this condition is also sufficient via induction on $\ell$. In the base case $\ell=1$, let $S \subseteq V_1=V$ be an arbitrary subset of $k/w_1$ vertices (note that by $k/w_1 \le k \le |V_1|$, such a set indeed exists). By construction, $S$ has weight $|S|w_1 = k$ and is closed, as $G$ cannot contain any edges.

Thus let us assume that the claim holds for all $\ell' \le \ell - 1$ and consider a Frobenius instance with $d':=\gcd(w_1,\dots, w_\ell) \divides k$. Let $d := \gcd(w_1,\dots,w_{\ell-1})$. We may assume that $d \notdivides k$; otherwise, already the Frobenius instance $G[V_1 \cup \dots \cup V_{\ell-1}]$ satisfies the assumption $\gcd(w_1,\dots, w_{\ell-1})\divides k$ and we obtain a closed set by inductive hypothesis.

Intuitively, we want to use the variables in $V_\ell$ to reach the target weight $k$ modulo $d$; then we can reduce to a simpler instance where every weight (including the target weight) is divided by $d$.
Note that we may assume 
\begin{equation}\label{eq:boundsond}
 2 \le d \le \sqrt{k/2},
 \end{equation}
where the lower bound follows from $d\notdivides k$ and the upper bound follows from $d\le \min_{i\in [\ell-1]} w_i \le \sqrt{k/2}$, as $w(v) \le w(D(v)) \le \sqrt{k/2}$ for any $v\in V$.

Let $b$ be the smallest non-negative integer such that $b \cdot  w_\ell \equiv k \pmod d$. Such an integer exists and satisfies $b < d$: By B\'ezout's identity, since $\gcd(w_\ell, d) = d' \divides k$, there are coefficients $\beta, \gamma$ such that $\beta w_\ell + \gamma d = k$, and thus any $b$ with $b \equiv \beta \pmod d$ achieves the desired congruence. 

Let $S\subseteq V_\ell$ be an arbitrary subset of size $b< d$; such a set indeed exists as $d\le \sqrt{k/2} \le k \le |V_\ell|$. We observe that $S$ satisfies
\begin{equation}\label{eq:boundonwDs}
w(D(S)) \le \sum_{s \in S} w(D(s)) \le |S| \sqrt{k/2} \le d\sqrt{k/2} \le \frac{k}{2},
\end{equation}
where we used~\ref{item:weights} for the second inequality, and~\eqref{eq:boundsond} for the last inequality. 
Consider the graph $G' = (V', E')$ obtained as a copy from $G$ from which we delete $V_\ell \cup D(S)$ and define the node weights $w'(v') = w(v)/d$ for any $v\in V\setminus (V_\ell \cup D(S))$. We claim that $G'$ is a Frobenius instance with parameter $k' := (k-w(D(S)))/d$ (observe that $k'$ is indeed integer, as $w(D(S)) \equiv b w_\ell \equiv k \pmod d$, and that $k'\ge 0$ by~\eqref{eq:boundonwDs}). If this is indeed the case, then by inductive hypothesis $G'$ has a closed set $S'$ with $w'(S') = k'$, since the gcd of the weights $w'$ is 1. Observe that by construction, $D(S)\cup S'$ is a closed set in $G$ of weight $w(D(S)) + d \cdot w'(S') = w(D(S)) + (k-w(D(S))) = k$, as desired.

It remains to prove that $G'$ is indeed a Frobenius instance with parameter $k'$. First, observe $G'$ has $\ell-1$ layers $V'_i := V_i \setminus D(S), i\in [\ell-1]$ and that $w'$ is well defined, as $d\divides w_i$ for all $i\in [\ell-1]$. Conditions~\ref{item:structweights} and~\ref{item:structedges} of being Frobenius are fulfilled as $G'$ is a subgraph of $G$. To see~\ref{item:size}, note that 
\[ |V'_i| \ge |V_i| - |D(S)| \ge |V_i|- w(D(S)) \ge k-w(D(S)) \ge k'. \]
To see~\ref{item:weights}, we observe that by~\eqref{eq:boundonwDs} and~\eqref{eq:boundsond}, we have
\[ k' = \frac{k-w(D(S))}{d} \ge \frac{k-k/2}{d} = \frac{k}{2d} \ge \frac{k}{d^2}. \]
Thus, for any $v'\in V'$, we obtain
\[ w'(D(v')) \le \frac{w(D(v))}{d} \le \frac{\sqrt{k/2}}{d} = \sqrt{\frac{k}{2d^2}} \le \sqrt{k'/2}, \]
where we used condition~\ref{item:weights} of $G$ in the second inequality. Thus, $G'$ is indeed a Frobenius instance with parameter $k'$, concluding the claim and thus the proof of our lemma.
\end{proof}

The above criterion is the main technical tool in the algorithmic result of the session. What remains is to show that the instance can be preprocessed in a way that it becomes a Frobenius instance.
\begin{thm}\label{thm:implalgo}
We can solve \wDAGImplication in time $f(k)n^{4\sqrt{k}}$.
\end{thm}
\begin{proof}
Consider the following recursive algorithm, which proceeds in 4 steps:

\emph{Step 1:} For every $v\in V$ with $w(D(v)) \ge \sqrt{k/2}$, we return YES if a recursive call determines that $G[V \setminus D(v)]$ has a closed set of weight $k-w(D(v))$; otherwise, we delete $v$ and all its ancestors from $G$. From now on, $G$ satisfies $w(D(v)) \le \sqrt{k/2}$ for all $v\in V$.

\emph{Step 2:} We construct layers $L_1, \dots, L_{\sqrt{k/2}}$ by the following iterative process: for every $i=1, \dots, \sqrt{k/2}$, we let $L_i$ consists of all vertices in $V\setminus (L_1 \cup \cdots \cup L_{i-1})$ whose outgoing edges end in $L_1\cup \cdots \cup L_{i-1}$. Note that $L_1, \dots, L_{\sqrt{k/2}}$ partitions $V$; in particular, every vertex is included in some $L_i$, since if there was a vertex $v\in V\setminus (L_1 \cup \cdots \cup L_{\sqrt{k/2}})$, then by construction there exists a path from $v$ containing strictly more than $\sqrt{k/2}$ vertices, leading to the contradiction $w(D(v))\ge |D(v)| > \sqrt{k/2}$. 

We observe that each layer $L_i$ can be partitioned into sublayers $L_{i,j}, j\in \{1, \dots, \sqrt{k/2}\}$ such that each $v\in L_{i,j}$ has weight $w(v) = j$: there can be no vertex of larger weight, as otherwise $w(D(v)) \ge w(v) > \sqrt{k/2}$ yields a contradiction. We consider layers $L_{i,j}$ in increasing lexicographic order of $(i,j)$: If $|L_{i,j}| < k$, then for every $v\in L_{i,j}$, we return YES if a recursive call determines that $G[V\setminus D(v)]$ contains a closed set of size $k-w(D(v))$, and otherwise we delete $v$ and all its ancestors from~$G$. Observe that by the lexicographic ordering, we never delete vertices from already processed layers, so that at the end of the process, each $L_{i,j}$ is either empty or contains at least $k$ vertices. %

\emph{Step 3:} We let $V_1, \dots, V_\ell$ be an enumeration of all non-empty sublayers $L_{i,j}$ by the lexicographic order on $(i,j)$ so that any vertex $v\in V_{i}$ has only edges to vertices in $V_1 \cup \cdots \cup V_{i-1}$. Observe that by construction, this yields a Frobenius instance. Let $w_1,\dots, w_\ell$ be the weights of the Frobenius instance. We return YES if $\gcd(w_1,\dots, w_\ell) \divides k$ and NO otherwise. 

Using Lemma~\ref{lem:Frobcrit}, the correctness of the algorithm is easy to see.

\begin{claim}\label{cl:imp1}
The above algorithm is correct.
\end{claim}
\begin{proof}
If the algorithm returns YES, indeed there is a closed set of size $k$: 
If we return YES in Steps 1 or 2, we have found a vertex $v$ and a closed set $S'$ in $G[V\setminus D(v)]$ of size $k-w(D(v))$, which yields a closed set $S'\cup D(v)$ in $G$ of size $k$, as desired. Otherwise, we have arrived at a Frobenius instance and returned YES since $\gcd(w_1,\dots, w_\ell)\divides k$, which implies that $G$ has a closed set of size $k$ by Lemma~\ref{lem:Frobcrit}.

Conversely, fix a closed set $S$ of size $k$, and we show that the algorithm returns YES:
If $S$ contains a vertex $v$ investigated in Steps 1 or 2, then the recursive call to $G[V\setminus D(v)]$ (for the first such vertex $v$) will find a solution of size $|S|-w(D(v))$ (note that $D(v) \subseteq S$ if $v\in S$). Otherwise, we have arrived at a Frobenius instance which must satisfy $\gcd(w_1,\dots, w_\ell)$ by Lemma~\ref{lem:Frobcrit}, and we return YES.
\cqed\end{proof}

Finally, we need to bound the running time of the recursive algorithm. The analysis relies on the observation that the algorithm makes at most $n$ recursive calls with a parameter decrease of at least $\sqrt{k/2}$, and at most $O(k^2)$ recursive calls with a parameter decrease of one.
\begin{claim}\label{cl:imp2}
The above algorithm can be implemented in time $f(k) n^{4\sqrt{k}}$.
\end{claim}
\begin{proof}
Let $U$ be the set of vertices of small layers ($|L_{i,j}| < k$) considered in Step 2. We observe that the above algorithm can be implemented recursively with the following recurrence on its running time $T(n,k)$ on instances with $n$ vertices and parameter $k$.
\begin{align*}
T(n,k) & \le \sum_{v\in V, w(D(v)) \ge \sqrt{k/2}} T(n, k-w(D(v))) + \sum_{u \in U} T(n,k-w(D(u))) + \Oh(n^2) 
\end{align*}
We claim by induction on $k$ that this yields a bound of $T(n,k) \le f(k) n^{4\sqrt{k}}$ for some $f(k) = k^{\Oh(k)}$. It is not difficult to see that for $k\le 2$, we can solve the problem in time $\Oh(n^2) = \Oh(n^{4\sqrt{k}})$, yielding the base case. For $k\ge 3$, we thus obtain the following bound, using that in Step 2, we process less than $k$ vertices for each ``small'' sublayer $L_{i,j}, 1\le i,j \le \sqrt{k/2}$, i.e., $|U|\le k (k/2) = k^2/2$,
\begin{align*}
T(n,k) & \le \Oh(n \cdot f(k-\sqrt{k/2}) n^{4\sqrt{k-\sqrt{k/2}}} + k^2 f(k-1) n^{4\sqrt{k-1}} + n^2) \\
   & \le (f(k)/2) (n^{4\sqrt{k-\sqrt{k/2}}+1} + n^{4\sqrt{k}}) \le f(k)n^{4\sqrt{k}},
  \end{align*}
 where the second bound follows from choosing $f(k) = k^{\Oh(k)}$ large enough to ensure $k^2f(k-1) \le f(k)/2$ and 
 the last bound follows from the observation that $4\sqrt{k-\sqrt{k/2}} + 1\le 4\sqrt{k}$ if and only if 
   \begin{align*}
& &  \left(4\sqrt{k-\sqrt{k/2}} + 1\right)^2 & \le 16k\\
& \iff & 16(k-\sqrt{k/2}) + 8\sqrt{k-\sqrt{k/2}} + 1 & \le 16k\\
& \iff & 8\sqrt{k-\sqrt{k/2}} + 1 & \le 16\sqrt{k/2},
  \end{align*}
  where the last inequality holds since $8\sqrt{k}+1 \le 16\sqrt{k/2}$ as $k\ge 3$.
  \cqed\end{proof}
Claims \ref{cl:imp1} and \ref{cl:imp2} show the correctness of our algorithm for \wDAGImplication. By the reduction described at the beginning of the section, a similar algorithm follows for \Implication.
\end{proof}
\section{Algorithms for $\NAND_2$-avoiding $\confam$: Reduction to Implication}
\label{sec:redtoImpl}

In this section, we show that for any $\NAND_{2}$-avoiding constraint family $\confam$, we can reduce $\SAT(\confam)$ to \Implication. Specifically, we obtain the following theorem.
\begin{thm}\label{thm:redtoImpl}
Let $\confam$ be a $\NAND_2$-avoiding constraint family and let $T_\Impl(n,k)$ denote the optimal running time to solve \Implication. There is a constant $c_\confam$ and computable $f(k)$ such that we can solve $\SAT(\confam)$ in time 
$f(k)( T_\Impl(n,k) + n^{c_\confam})\log n$.
\end{thm}

Together with Theorem~\ref{thm:implalgo}, this gives an $f(k)n^{4\sqrt{k} + c_\confam}$-time algorithm for any $\NAND_2$-avoiding constraint family $\confam$.

To prove the above theorem, we prepare some notation and helpful facts. Let $\phi$ be an arbitrary formula.  For any assignment $a$, we call $a'$ a \emph{minimal satisfying extension} of $a$, if $a'$ satisfies $\phi$, $a\le a'$, and no other satisfying assignment $a''\notin \{a,a'\}$ fulfills $a\le a'' \le a'$. The following lemma shows that there are only $f(k)$ many minimal extensions of weight at most $k$, and these minimal extensions can be computed in time $f(k) n^{c}$ for some constant $c$ independent of $k$. Intuitively, this follows by using the bounded search tree technique over violated constraints, where the depth of the search tree is bounded by $k$ and each branching step has at most $r$ possibilities.

\begin{lem}[{\cite[Lemma 2.3]{BulatovM14}}]\label{lem:minextensions}
Let $\confam$ be a finite constraint family of bounded arity $r$. There is a constant $c_\confam'$ such that given any instance $\phi$ of $\SAT(\confam)$ and assignment $a$, there are at most $\Oh(r^k)$ minimal extensions of $a$ of weight $k$, and we can compute these extensions in time $\Oh(r^k n^{c'_\confam})$. 
\end{lem}

As an immediate useful consequence, we obtain that for our algorithmic results, we may assume without loss of generality that $\confam$ is $0$-valid, i.e., each $f\in \confam$ is satisfied by the all-zeroes assignment.
\begin{cor}[{see also \cite[Lemma 4.1]{Marx05}}]\label{cor:wlog0valid}
We can reduce any instance of $\SAT(\confam)$ with parameter~$k$ to $\Oh(r^k)$ many instances of $\SAT(\confam')$ with a parameter bounded by $k$, where $\confam'$ is the set of all $0$-valid $f'$ that are represented by $\confam$.
\end{cor}
\noindent 
By definition, if $\confam$ does not represent $\NAND_2$, then also $\confam'$ does not represent $\NAND_2$, and it remains to give an $f(k)(T_\Impl(n,k) + n^{c_\confam})\log n$-time algorithm for \emph{$0$-valid} $\NAND_2$-avoiding $\confam$.  

In the remainder of this section, we will use the graph formulation of the \Implication problem: We are given a directed graph $G=(V,E)$ and the task is to find a closed set $S$ (recall that $S$ is closed, if for all $(u,v)\in E$ we have that $u\in S$ implies $v\in S$) of size $k$. Recall that for any vertex set $S\subseteq V$, $D(S)$ denotes the set of descendants of any vertex $s \in S$ (including the vertices in $S$).

Our aim is the following: Given a formula $\phi$ of $\SAT(\confam)$, we give a randomized construction of an \Implication instance~$G$ such that 
\begin{enumerate}[label=(\roman*),nosep]
\item any closed set $S$ in $G$ corresponds to a satisfying assignment of $\phi$, and 
\item with large enough probability, $G$ contains a closed set of size $k$ if $\phi$ has a weight-$k$ solution.
\end{enumerate}
To this end, we let $V=\{x_1, \dots, x_n\}$ and recall that, for any set $S\subseteq V$, we let $a_S: [n]\to \{0,1\}$ denote a corresponding assignment with $a_S(i) = 1$ iff $x_i \in S$. From now on, we often synonymously speak of closed sets $S\subseteq V$ in $G$ and the corresponding assignment $a_S$ for $\phi$.

The rough outline is as follows: we start with the graph $G=(V,\emptyset)$, and try to repeatedly ``fix'' some closed set $S$ that violates $\phi$, by determining a (random) implication consistent with a minimal satisfying extension of $S$. The main insight is that if $\confam$ avoids $\NAND_2$, then it suffices to make sure that all sets $D(v)$ for $v\in V$ are satisfying and this will automatically ensure that every closed set is satisfying.

Let us formally describe the algorithm: %
\begin{enumerate}
\item Given $\phi$, initialize $G= (V,E)$ with $V=\{x_1,\dots, x_n\}$ and $E=\emptyset$.
\item While there exists some $v\in V$ such that $a_{D(v)}$ violates $\phi$, do the following:
\begin{enumerate}
\item Compute the set $A_v$ of minimal satisfying extensions of $a_{D(v)}$ of weight at most $k$.
\item Let $X$ consist of all $x_i \in V\setminus D(v)$ such that there is some $a\in A_v$ with $a(i) = 1$.
\item If $X = \emptyset$, delete all ascendants of $v$ (including $v$) from $G$. Otherwise, pick $x$ uniformly at random from $X$ and add the edge $(v,x)$ to $E$.
\end{enumerate}
\end{enumerate}

The important properties of the algorithm are captured in the following lemma.
\begin{lem}\label{lem:randprop}
Let $\confam$ be a finite $0$-valid constraint family. There is a constant $c_\confam$ and a function $g(k)$ such that the following properties hold.
\begin{enumerate}[label=(P\arabic*)]
\item\label{item:iters} During the process, each vertex $v$ is considered at most $k$ times in the while loop. Thus, the algorithm can be implemented to run in time $\Oh(g(k)n^{c_\confam})$.
\item\label{item:satisclosedwhp} If $\phi$ has a satisfying assignment of weight $k$, then with probability at least $g(k)^{-1}$, there is a closed set $S$ in $G$ of size $k$.
\item\label{item:closedissat} If $\confam$ avoids $\NAND_2$, any closed set $S\subseteq V$ in the constructed graph yields a satisfying assignment $a_S$ for $\phi$. 
\end{enumerate}
\end{lem}
\begin{proof}
For \ref{item:iters}, note that whenever $v\in V$ is considered in the while loop, it is either deleted, or an edge $(v,x)$ with $x\notin D(v)$ is added to the graph. Thus, when $v$ is considered for the $k$-th time, we have $|D(v)|\ge k$, and thus there can be no satisfying extension of $a_{D(v)}$ of weight at most $k$. Consequently, we must have $A_v=\emptyset$, and thus $X=\emptyset$, which forces $v$ to be deleted. Thus, we have at most $kn$ iterations of the while loop, where each iteration can be implemented in time $\Oh(r^k n^{c'_\confam})$ by Lemma~\ref{lem:minextensions}.

For~\ref{item:satisclosedwhp}, assume that there is a set $S$ of size $k$ such that $a_S$ satisfies $\phi$. We show that with large enough probability, we will maintain as invariant that $D(v)\subseteq S$ for every $v\in S$, and thus $S$ will be a closed set in $G$. To this end, we first observe that for $D(v) \subseteq S$ to hold for all $v\in S$, it suffices that the following property holds:
\begin{equation}\label{eq:processgood}
\text{In each iteration that considers a vertex $v\in S$, the selected vertex $x$ is in $S$.}
\end{equation}
Indeed, if this is the case, then no $v\in S$ is ever deleted. Furthermore, we have that $D(v)\subseteq S$ for all $v\in S$, and thus $S$ is a closed set in $G$. It remains to give a lower bound on the probability that~\eqref{eq:processgood} holds throughout the process.

To this end, consider the event that some $v\in V$ is considered in the while loop, conditioned that~\eqref{eq:processgood} has not been violated in a previous iteration. Under this event, $D(v) \subseteq S$, and thus there is a minimal satisfying extension $D(v) \subsetneq S' \subseteq S$ such that $a_{S'}$ satisfies $\phi$ and thus $a_{S'}\in A_v$. Let $s\in S'\setminus D(v)$ be arbitrary, then $s\in X$ by construction (note that $s$ has not been deleted). By Lemma~\ref{lem:minextensions}, we have that $|A_v| \le \Oh(r^k)$. Since each $a\in A_v$ has weight at most $k$, this yields $|X|\le k|A_v|\le \Oh(k r^k)$. Thus, the probability that the random choice is $x=s$ is at least $1/|X|\ge \Omega(1/(kr^k))$. Finally, we observe that by~\ref{item:iters}, for each $v\in S$, there are at most $k$ iterations considering $v$, where each iteration has a probability of at least $\Omega(1/(kr^k))$ of not violating~\eqref{eq:processgood}. Thus, we obtain that~\eqref{eq:processgood} holds with probability at least $\Omega(1/(kr^k)^{k|S|}) = \Omega(1/(kr^k)^{k^2})$, and the claim follows by setting $g(k) \coloneqq (kr^k)^{-k^2}$.

Finally, for \ref{item:closedissat}, note that at the end of the process, the property holds that
\begin{equation}\label{eq:Nvclaim}
\text{For all (remaining) } v\in V, a_{D(v)} \text{ satisfies } \phi.
\end{equation}
We will leverage this fact to show that $a_S$ satisfies $\phi$ for \emph{all} closed sets $S=D(v_1)\cup ... \cup D(v_\ell)$ for $v_1,\dots, v_\ell\in V$. We first transform the graph $G$ to a DAG by contracting all strongly connected components $C=\{v_1, \dots, v_c\}$ to a single vertex $v_C$ representing the set $C$. Note that the closed sets in the DAG remain in a one-to-one correspondence to the closed sets of the original graph (and the corresponding assignments to $\phi$), thus this transformation is without loss of generality. Thus, we may assume that $G$ has a topological ordering $v_1,\dots, v_{n'}$ of its vertices ($n'\le n$). We will prove by induction on $i=n', ..., 1$ that for all closed sets $S\subseteq \{v_i,...,v_{n'}\}$, $a_S$ satisfies $\phi$.

For the base case $i=n'$, we only need to verify that (i) the all-0 assignment satisfies $\phi$, which holds by 0-validity of $\confam$, and (ii) that $a_{v_{n'}}$ satisfies $\phi$, which holds by~\eqref{eq:Nvclaim} (as $D(v_{n'}) = \{v_{n'}\}$). Thus, for $i<n'$, let us assume that the claim holds for $i+1$. Consider any closed set $U\subseteq \{v_i, \dots, v_{n'}\}$. If $U$ does not contain $v_i$, the claim follows by inductive assumption, thus let us assume that $v_i \in U$ and thus $U\supseteq D(v_i)$, as $U$ is closed. If $U = D(v_i)$, $a_U$ satisfies $\phi$ by~\eqref{eq:Nvclaim}. Thus, it remains to consider $U \supsetneq D(v_i)$, for which we assume for contradiction that $a_U$ violates~$\phi$. Let $W := U\setminus D(v_i)$, and note that $D(W)\subseteq U$ is a closed set in $\{v_{i+1},\dots, v_{n'}\}$. Thus, by inductive assumption, $a_{D(W)}$ satisfies $\phi$. Furthermore, observe that $Z:=D(v_i)\cap D(W)$ is a closed set in $\{v_{i+1}, \dots, v_{n'}\}$ (since the intersection of any two closed sets yields a closed set). Thus, $a_Z$ satisfies $\phi$ by inductive assumption. It remains to show that the fact that $a_{D(v_i)}$, $a_{D(W)}$ and $a_Z=a_{D(v_i)\cap D(W)}$ all satisfy $\phi$, while $a_U=a_{D(v_i)\cup D(W)}$ violates $\phi$, gives a contradiction to $\confam$ avoiding $\NAND_2$.

To this end, let $C$ be a constraint violated by $a_U$ and note that $C=f(x_{i_1},\dots,x_{i_r})$ for some $f\in \confam$ and $i_1,\dots, i_r\in [n]$. Note that we can view $f$ as $f: \{0,1\}^{V_c} \to \{0,1\}$ for some appropriate variable set $V_C$. We show how to obtain $\NAND_2$ as a restriction of $f$ by partitioning $V_C$ into $X'\coloneqq (D(v_i) \setminus Z) \cap V_C, Y' \coloneqq (D(W)\setminus Z)\cap V_C, Z_1 \coloneqq Z\cap V_C, Z_0 \coloneqq V_C \setminus (X'\cup Y' \cup Z_1)$ and observing that 
\begin{align*}
\begin{matrix}
f( \overbrace{\entry{0}}^{X'},& \overbrace{\entry{0}}^{Y'}, & \overbrace{\entry{0}}^{Z_0}, & \overbrace{\entry{1}}^{Z_1}) & = & 1, & \text{[since $a_Z$ satisfies $C$]} \\
f( \entry{1},& \entry{0}, & \entry{0}, & \entry{1}) & = & 1, & \text{[since $a_{D(v_i)}$ satisfies $C$]} \\
f( \entry{0},& \entry{1}, & \entry{0}, & \entry{1}) & = & 1, & \text{[since $a_{D(W)}$ satisfies $C$]} \\
  \hline
f( \entry{1},& \entry{1}, & \entry{0}, & \entry{1}) & = & 0. & \text{[since $a_{D(W)\cup D(v_i)}$ violates $C$]} \\
\end{matrix}
\end{align*}
\end{proof}

It remains to give the proof of Theorem~\ref{thm:redtoImpl}.
\begin{proof}[Proof of Theorem~\ref{thm:redtoImpl}]
By Corollary~\ref{cor:wlog0valid}, we may assume without loss of generality that $\confam$ is 0-valid.
We repeat the following process $g(k)$ many times: We use the above algorithm to generate an \Implication instance $G$, and return YES if $G$ contains a closed set of size $k$, which we determine using an optimal \Implication algorithm. If none of the $g(k)$ iterations were successful, we return NO. Note that this approach can be implemented in time $g(k)\Oh(g(k)n^{c_\confam} + T_\Impl(n,k))$ by~\ref{item:iters}, and correctly decides the instance with probability at least $1-(1-1/g(k))^{g(k)} \ge 1-1/e$ by~\ref{item:satisclosedwhp} and~\ref{item:closedissat}.

The algorithm described above can be derandomized using the standard technique of Color Coding \cite{AlonYZ95}. In each iteration when vertex $v$ is considered, a random vertex $x$ is selected from a set $X$ of at most $K=\Oh(kr^k)$ vertices. As each vertex is considered at most $k$ times, we can represent the random choices by a function $r:V\to [K]^k$,  with the meaning that $r(v)$ is the vector of choices made when considering vertex $v$.  As discussed in the proof of Lemma~\ref{lem:randprop}, when considering vertices $v\in S$, these random choices need to be consistent with $S$ to ensure that $S$ is a closed set in the resulting graph. That is, for each $v\in S$ there is a vector $c(v)\in [K]^k$ such that if the random choice satisfies $r(v)=c(v)$ for every $v\in S$, then $S$ is a closed set.

We say that a family $\mathcal{H}$ of functions $h:[n]\to [k]$ is a {\em $(n,k)$-perfect family of hash functions} if for every $S\subseteq V$ of size $k$, there is an $h\in \mathcal{H}$ that is injective on $S$, i.e., assigns different values to different elements of $S$. It is known that a $(n,k)$-perfect family of size $2^{O(k)}\log n$ can be computed in time $2^{O(k)}n\log n$ \cite{AlonYZ95}. The derandomized algorithm would first compute such a family $\mathcal{H}$ over $V$ and would iteratively go through every $h\in \mathcal{H}$ and function $q:[k]\to [K]^k$. For a given choice of $h$ and $q$, we define the function $r(v)=q(h(v))$ and run the randomized algorithm using this function $r$ instead of the random choices. It is easy to see that the definition of $(n,k)$-perfect hash functions implies that there is at least one choice of $h$ and $q$ where $r(v)$ is exactly the prescribed value $c(v)$ for every $v\in S$ and therefore the randomized algorithm correctly finds the solution $S$. As we are considering at most $|\mathcal{H}|=2^{O(k)}\log n$ functions $h$ and $K^{k^2}$ different functions $q$, there is a function $f(k)$ such that the total running time is at most $f(k)\log n$ times a single run of the randomized algorithm.
\end{proof}

\section{Algorithms for $\NAND$-representing $\confam$: Reduction to Clique}
\label{sec:redToClique}

In this section, we develop algorithm for constraint families that might represent $\NAND_2$, but avoid $\NAND_d$ for some $d\ge 3$. To this end, we give a reduction to $(d-1)$-uniform HyperClique for $\NAND_d$-avoiding families, giving in particular a $f(k)n^{(\omega/3)k+c_\confam}$-time algorithm for $\NAND_3$-avoiding families. 

We first start with a natural reduction of $\SAT(\confam)$ for any $\confam$ with arity bounded by $r$ to $r$-uniform HyperClique, based on color-coding. To this end, let $\THC{d}(n,k)$ denote the optimal running time of finding a $k$-clique in a $d$-uniform hypergraph.
\begin{prop}\label{prop:aritybaseline}
Let $\confam$ be a constraint family of arity at most $r$. Then $\SAT(\confam)$ can be solved in time $f(k)(n^{2r} + \THC{r}(n,k))\log n$.
\end{prop}
\begin{proof}
Let $\phi$ be an arbitrary $\SAT(\confam)$ formula. Observe that any constraint $C$ of $\phi$ depends only on a set $\vars(C)\subseteq \{x_1,\dots,x_n\}$ of at most $r$ variables. For an assignment $a$, we let $C(a)\in \{0,1\}$ denote whether $C$ is satisfied by $a$.

We first show how to determine, given a partition of $x_1, \dots, x_n$ into $k$ sets $X_1,\dots, X_k$, whether there is a solution that sets precisely one variable in each $X_i$ to true. To this end, we construct a hypergraph $G$ with vertex set $X_1 \cup X_2 \cup \cdots \cup X_k$ and the following set of hyperedges: we include each possible hyperedge $e=\{x_{j_1}, \dots, x_{j_r}\}$ with $x_{j_1}\in X_{j_1}, \dots, x_{j_r} \in X_{j_r}$ and distinct $j_1, \dots, j_r\in [k]$ \emph{unless} there exists a clause $C$ with $\vars(C) \subseteq X_{j_1} \cup \cdots \cup X_{j_r}$ which is violated by the assignment that sets precisely the variables $e=\{x_{j_1}, \dots, x_{j_r}\}$ to $1$, i.e., $C(a_{e}) = 0$.

We claim that $H := \{x_{i_1}, \dots, x_{i_{k}}\}$ with  $x_{i_1} \in X_1, \dots, x_{i_k} \in X_k$ yields a $k$-clique in $G$ if and only if the assignment $a_{H}$ satisfies $\phi$. Indeed, assume that there is a clause $C$ violated by $a_{H}$. Note that as $C$ has arity at most $r$, we have $\vars(C) \subseteq X_{j_1} \cup \cdots \cup X_{j_r}$ for some distinct $i_1, \dots, i_r \in [k]$ (if $C$ involves variables of less than $r$ sets, we may use arbitrary additional sets). Thus, $e:=\{x_{j_1}, \dots, x_{j_r}\}$ cannot be an edge in $G$, since $a_H$ violates $C$, $a_e$ and $a_H$ agree on $\vars(C)$, and thus also $a_{e}$ violates~$C$. Conversely, if there is some $e:= \{x_{i_1}, \dots, x_{i_r}\}$ with distinct $i_1, \dots, i_r \in [k]$ such that $e$ is not an edge in $G$, then there exists some clause $C$ with $\vars(C) \subseteq X_{i_1} \cup \dots \cup X_{i_r}$ which is violated by $a_e$. Since $a_H$ and $a_e$ agree on $\vars(C)$, we conclude that also $a_H$ violates $C$ and thus $\phi$.

To create the desired $k$-partition of variables, we use a (deterministic) color-coding scheme: Let $\mathcal{H}$ be a \emph{$(n,k)$-perfect family of hash functions} $h : [n]\to [k]$ -- recall that this means that for any $S=\{s_1,\dots, s_k\}\subseteq [n]$, there exists some $h\in \mathcal{H}$ such that $\{h(s_1),\dots, h(s_k)\} = \{1,\dots, k\}$. Known efficient constructions~\cite{SchmidtS90, AlonYZ95} produce such assignments with $\ell = 2^{\Oh(k)}\log(n)$ in time $2^{\Oh(k)}n \log n$. Given this family, we create for each $h \in \mathcal{H}$ the $k$-partition $X^{(h)}_1, \dots, X^{(h)}_k$ with $X^{(h)}_j = \{ x_{s} \mid h(s) = j\}$ and solve the corresponding $r$-uniform HyperClique instance in time $\THC{r}(n,k)$. If any of these instances returns a solution, then indeed $\phi$ has a satisfiable assignment of weight $k$. Conversely, if $a_S$ is a weight-$k$ satisfying assignment for $\phi$, then by construction, there exists a hash function $h \in \mathcal{H}$ such that $|S\cap X^{(h)}_j|=1$ for $j=1,\dots, k$, and thus the corresponding $r$-uniform HyperClique instance indeed contains a solution. For each of the $2^{\Oh(k)}\log(n)$ hash functions, the time to construct and solve the $d$-uniform HyperClique instance is bounded by $\Oh(n^{2r} + \THC{r}(n,k))$, concluding the claim. 
\end{proof}

The main result in this section is the following reduction from $\NAND_{d+1}$-avoiding constraint families to $d$-uniform HyperClique.

\begin{thm}\label{thm:NANDavoidingalgo}
Let $d\ge 2$ and $\confam$ be an $\NAND_{d+1}$-avoiding constraint family. If there are constants $\gamma \ge d/(d+1)$ and $c$, and a computable $g(k)$ such that $d$-uniform HyperClique can be solved in time $g(k) n^{\gamma k + c}$, then there is a constant $c'$ and computable $g'(k)$ such that $\SAT(\confam)$ can be solved in time $g'(k) n^{\gamma k + c'}$.
\end{thm}
In particular, since we can find $k$-cliques in graphs in time $\Oh(n^{\frac{\omega}{3}k + 1})$, we obtain an $g(k) n^{\frac{\omega}{3}k + c'}$-time algorithm for solving $\SAT(\confam)$ for all $\NAND_3$-avoiding constraint families. Similarly, if for $d\ge 3$ the $d$-uniform HyperClique conjecture is refuted by exhibiting a $g(k)n^{(1-\eps)k+c}$-time algorithm for some constants $0< \eps < 1/(d+1)$ and $c$, we would obtain a $g'(k) n^{(1-\eps)k+c'}$-time algorithm for $\SAT(\confam)$ for $\NAND_{d+1}$-avoiding families $\confam$.

In the remainder of the section, we give the proof of Theorem~\ref{thm:NANDavoidingalgo}. 
The main task of the algorithm is to detect \emph{robust} assignments, defined as follows.

\begin{defn}
Let $a: [n] \to \{0,1\}$ be a weight-$k$ assignment that satisfies $\phi$. We say that $a$ is \emph{$d$-robust} if there is no assignment $a' \le a$ of weight at most $d$ that violates $\phi$.
\end{defn}

The first step of the algorithm is the easier task of detecting satisfying assignments that are \emph{not} $d$-robust (if there exists any):
Intuitively, an assignment that is not $d$-robust offers an advantage to find it: Assume we correctly guess an assignment $a'\le a$ of weight $w \le d$ such that some clause $C$ is violated by $a'$, then to extend $a'$ to the satisfying assignment $a$, we know that at least one additional variable in $C$ must be set to true. By bruteforcing over the at most $r-w \le r$ many possibilities, we gain an advantage. Specifically, by enumerating $O(n^{w}r) = \Oh(n^{w})$ many possibilities, we can fix $w+1$ true variables in our solution.

Let $T(n,k)$ denote the time our algorithms takes to solve an arbitrary $\SAT(\confam)$ instance for a $\NAND_{d+1}$-avoiding family $\confam$. 
In a preprocessing step, we first enumerate all assignments $a'$ of weight at most $d$. If there exists a clause $C_{a'}$ that is violated by $a'$, then we enumerate all variables $x \in \vars(C_{a'})\setminus \ones(a')$ (recall that $\vars(C)$ is the set of variables involved in $C$ and $\ones(a)$ denotes the set of variables set to $1$ under $a$). We recursively determine satisfiability of the formula $\phi_{a',x}$ obtained by restricting all variables in $\ones(a') \cup \{x\}$ to true. Disregarding the time to determine existence of violated clauses $C_{a'}$, this step takes time 
\begin{equation}\label{eq:rec1}
\sum_{w=0}^d \sum_{\substack{\text{weight-}w\\ \text{assignment }a'}} \sum_{x\in \vars(C_{a'})\setminus \ones(a')} T(n, k-(w+1)) \le \sum_{w=0}^d O(n^w) T(n, k-(w+1)). 
\end{equation}
To determine a violated clause $C_{a'}$ (if it exists) for all weight-($\le d$) assignments $a'$, we simply traverse each clause $C$, determine the at most $\sum_{w=0}^d {r \choose w}= O(1)$ weight-($\le d$) assignments violating $C$ and store $C$ as violated for each of these assignments (if no other clause is already stored). This step takes time $O(m) = O(n^r)$ in the beginning.

After this preprocessing, it remains to consider $d$-robust assignments. To determine whether a $d$-robust assignment satisfies $\phi$, we define a formula $\phi_d$ that is satisfied only by satisfying assignments of $\phi$, and particularly by all $d$-robust satisfying assignments of $\phi$. To this end, let $F_d$ contain all assignments of weight at most $d$ that violate some clause $C$ of $\phi$, and define
\[ \phi_d := \bigwedge_{a \in F_d} \NAND(\ones(a)). \]
  \begin{lem}\label{lem:NANDrep}
  The constructed formula $\phi_d$ has the following properties:
\begin{enumerate}[label=(P\arabic*)]
  \item If $\confam$ is $\NAND_{d+1}$-avoiding, then any satisfying assignment $a$ of $\phi_d$ is a satisfying assignment of $\phi$.
  \item If $a$ is a $d$-robust satisfying assignment of $\phi$, then $a$ satisfies $\phi_d$.
    \end{enumerate}
  \end{lem}
  \begin{proof}
  To prove (P1), we will make use of the following property.
  \begin{prop}\label{prop:NANDd1}
  Let $\confam$ be a $\NAND_{d+1}$-avoiding family. Then if an assignment $a$ violates some clause $C$ (chosen from $\confam$), there is an assignment $a' \le a$ of weight at most $d$ that violates~$C$.
  \end{prop}
  \begin{proof}
  We prove the claim via induction on the weight $w$ of the clause $C$ under $a$. If $w \le d$, the claim trivially holds. To prove the inductive step, we may assume for contradiction that there is an assignment $a$ of weight $w\ge d+1$ violating some clause $C=f(\bar{x})$, but no assignment $a' \le a$ of weight at most $w-1$ violates $C$. We will show that $\NAND_{d+1}$ can be obtained as a restriction of $f$. 
  To this end, choose some set $S \subseteq \ones(a) \cap \vars(C)$ of size $d+1$ (which is possible as $w\ge d+1$), and partition $\vars(C)$ into $S$, $Z_1 \coloneqq (\ones(a) \cap \vars(C)) \setminus S$ and $Z_0 \coloneqq \vars(C) \setminus (S\cup Z_1)$. Observe that we have   
  \begin{align*}
  \begin{matrix}
f( \overbrace{y_1\dots y_{d+1}}^{S},& \overbrace{\entry{0}}^{Z_0}, & \overbrace{\entry{1}}^{Z_1}) & = & 1, & \quad \text{if $(y_1,\dots,y_{d+1})\ne (1,\dots, 1)$}\\
  \hline
f( \; \entry{1}\quad,& \entry{0}, & \entry{1}) & = & 0. & \quad\text{[since $a$ violates $C$]}
    \end{matrix}
    \end{align*}
    where the first line follows since no assignment $a' \le a$ of weight at most $w-1$ violates $C$, yielding a contradiction.
  \cqed \end{proof}
  To prove (P1), assume that an assignment $a$ violates some clause $C$ of $\phi$. Since $\confam$ is $\NAND_{d+1}$-avoiding, by Proposition~\ref{prop:NANDd1} there exists an assignment $a'\le a$ of weight at most $d$ such that $a'$ violates $C$. Thus, $\phi_d$ contains a clause $\NAND(\ones(a'))$, which is violated by $a$, as $a' \le a$.

  To prove (P2), assume for contradiction that a $d$-robust assignment $a$ satisfies $\phi$ but not $\phi_d$. Then there is some $a' \in F_d$ such that $\NAND(\ones(a'))$ is violated by $a$, i.e., $a' \le a$. As $a' \in F_d$, there must be a clause $C$ of $\phi$ that is violated by $a'\le a$, which proves that $a$ is not $d$-robust and thus yields a contradiction.
  \end{proof}

  Note that $\phi_d$ is a $\SAT(\confam')$ formula with constraint family $\confam' = \{ \NAND_j \mid 2 \le j \le d\}$ of arity~$d$. Thus, by Proposition~\ref{prop:aritybaseline}, we can determine satisfiability of $\phi_d$ in time $f(k) (n^{2d} + \THC{d}(n, k))\log n$. We obtain the following recurrence by combining~\eqref{eq:rec1}, the $O(m)$-time preprocessing to determine violated classes $C_{a'}$, and $f(k)(n^{2d} + \THC{d}(n,k))\log n$ to solve $\phi_d$:  
  \begin{equation}\label{eq:rec}
  T(n,k) =  \Oh(m) + f(k)(n^{2d} + \THC{d}(n,k))\log n + \sum_{w=0}^d O(n^w) T(n, k-(w+1))
  \end{equation}

  Assume that there are $\gamma \ge d/(d+1)$ and $c$ such that $\THC{d}(n,k) \le g(k)n^{\gamma k+c}$. We will show that  $T(n,k) = g'(k)\Oh(n^{\gamma k + c'})$ for any $c'>\max\{c, 2r\}$ and $g'(k) = f(k)g(k)$. 

  We prove the claim via induction on $k$. The base case is $k < c'$, in which case we can solve $\SAT(\confam)$ in time $f(k)(n^{2r} + \THC{r}(n,k))\log n = f(k)\Oh((n^{2r}+n^k)\log n) \le \Oh(n^{c'})$, satisfying the claim. Thus, let us assume that $k \ge c'$ and that the claim holds for all $k'\le k-1$. Using~\eqref{eq:rec}, we obtain
\begin{align*}
T(n,k) & \le \Oh(m) + f(k)(n^{2d} + g(k) n^{\gamma k + c})\log n + g'(k) \left(\sum_{w=0}^d O(n^w) n^{\gamma (k-w+1) + c'}\right) \\
    & \le g'(k) \log n\cdot \Oh\left(n^{2r} + n^{\gamma k + c} + \sum_{w=0}^d n^{w + \gamma(k-(w+1)) + c'}\right) \\
    & \le g'(k) \log n\cdot \Oh\left(n^{2r} + n^{\gamma k + c} + n^{\gamma k + c'}\right) = g'(k) \Oh(n^{\gamma k + c'}),
\end{align*}
where in the second line, we used that $g'(k) = f(k)g(k)$, and in the last line we used that $\gamma(w+1) \ge w$ as $\gamma \ge d/(d+1) \ge w/(w+1)$ for $w\le d$, as well as our choice of $c'$ which satisfies $c' > c$ and $\gamma k + c' \ge c' > 2r$.

\section{Hardness Results}
\label{sec:hardness}

In this section, we give our hardness results. To this end, we first consider $\Implication = \SAT(\Impl)$ and give a $f(k)n^{(\omega/6-o(1))\sqrt[3]{k}}$-lower bound under the $k$-clique conjecture. Afterwards, we handle the case of $\NAND_{d}$- or $\Impl$-representing families, by reducing from $d$-uniform (Hyper)Clique or \Implication, respectively.

\subsection{Hardness for \Implication}

\begin{thm}\label{thm:implhardness}
If \Implication can be solved in time $f(k) n^{(\omega/6-\eps)\sqrt[3]{k}+c}$ for some $\eps>0, c$ and $f(k)$, then the $k$-Clique conjecture fails.
\end{thm}
\begin{proof}

Let $G=(V, E)$ be an undirected graph.
We construct an \wDAGImplication instance $G' = (V', E', w')$ with parameter $k' = k\cdot K + {k \choose 2}$ with $K\coloneqq {k \choose 2}+1$ as follows.
The vertex set $V'$ is the disjoint union of \emph{vertex nodes} $V'_V \coloneqq \{v_u \mid u\in V\}$  and \emph{edge nodes} $V'_E \coloneqq \{v_e \mid e\in E\}$. For every $e=\{u,w\}\in E$, we introduce the edges $(v_e, v_u), (v_e, v_w)$ to $E'$. Furthermore, we set the weights of vertex nodes to $K$, and the weights of edge nodes to $1$. 

\begin{claim}
There is a closed set $X$ of weight $k'$ in $G'$ if and only if there is a $k$-clique in $G$.
\end{claim}
\begin{proof}
Let $C=\{v_1, \dots, v_k\}$ be a $k$-clique in $G$. Observe that $X = \{ v_u \mid u\in C\} \cup \{ v_e \mid e\in {C\choose 2}\}$ is a closed set in $G'$ of weight $|C|K + {|C| \choose 2} = k\cdot K + {k \choose 2} = k'$.

For the converse, assume that $X$ is a closed set in $G'$ of weight $k'$. Setting $X_V := X\cap V'_V$ and $X_E := X\cap V'_E$, we show the following sequence of facts:
\begin{enumerate}[label=\arabic*)]
\item $X_E \subseteq {X_V \choose 2}$: note that $X$ is only closed if for all $v_{\{u,w\}} \in X_E$, we have $v_u, v_w \in X_V$.
\item $|X_V| = k$ and $|X_E| = {k \choose 2}$: note that if $|X_V|< k$, then $|X_E| \le {k-1 \choose 2}$ by 1) and thus the weight of $X$ is $|X_V|K+|X_E| \le (k-1)K + {k-1 \choose 2} < kK+ {k \choose 2} = k'$. Furthermore, if $|X_V|>k$, then the weight of $X$ is at least $|X_V|K \ge (k+1)K = kK + {k \choose 2} + 1> k'$. Thus, we have $|X_V| = k$, and hence we must have $|X_E| = {k \choose 2}$ for $|X_V|K + |X_E| = k'$ to hold.
\item $X_V$ forms a $k$-clique in $G$: Facts 1) and 2) require that $X_E = {X_V \choose 2}$, which implies that $E$ contains all edges between vertices of $X_V$. 
\end{enumerate}
The last statement concludes the proof of the claim.
\cqed
\end{proof}

Assume that for some $c$ and $\eps >0$, there is an \Implication algorithm running in time $f(k)n^{(\omega/6-\eps)\sqrt[3]{k}+c}$. Given a $k$-clique instance $G$, we run the above reduction to create a \wDAGImplication instance $G'$ with parameter $k'\le (k+1)({k \choose 2} + 1) = (k^3 + k+2)/2 \le k^3$ for $k\ge 2$. Observe that $G'$ has $\Oh(n^2)$ nodes and can be converted to an equivalent \Implication instance $G''$ with the same parameter $k'$ and $\Oh(k^2n^2)$ nodes by simulating each node weight~$w$ by a cycle of $w$ nodes. Now, we determine whether  $G''$ has a closed set of weight $k' \le k^3$ using the \Implication algorithm and thus decide $k$-clique in time $f(k^3)\Oh((k^2n^2)^{(\omega/6-\eps) k + c})=f(k^3)k^{\Oh(k)}n^{(\omega/3-2\eps) k + 2c}$, refuting the $k$-Clique conjecture.
\end{proof}

\subsection{Hardness for $\SAT(\confam)$}

In this section, we give our hardness results for general constraint families $\confam$ by reducing from ($d$-uniform Hyper-)Clique either via the independent set problem or via \Implication.

To obtain these results, we frequently have to plug-in constant $0$s or $1$s to obtain our desired constraints. Technically, this is a non-trivial step, as we need to enforce some variables to be assigned fixed values without blowing up the number of variables or the weight of the desired solution. To facilitate our proofs, we first formalize the problem variant that allows us to plug-in constants freely.
\begin{defn}
Let $\confam$ be an arbitrary constraint family and $\Sigma\subseteq \{0,1\}$. The problem $\SAT_\Sigma(\confam)$ asks to determine whether a given formula $\phi$ with Boolean variables $x_1,\dots, x_n$ is has a satisfying assignment of weight $k$, where $\phi$ is a conjunction of $m$ constraints of the form $f(\mathbf{x})$, where $f:\{0,1\}^r \to \{0,1\}$ is a constraint function in $\confam$ and $\mathbf{x}$ is an $r$-tuple over $\{x_1,\dots, x_n\} \cup \Sigma$ (any variable or constant $c\in \Sigma$ may be used repeatedly). Note that $\SAT_\emptyset(\confam) = \SAT(\confam)$.
\end{defn}

Ideally, we would like to show that $\SAT_{\{0,1\}}(\confam)$ is equivalent to $\SAT(\confam)$. More specifically, we would like to employ reductions of the following form. 
\begin{defn}
Let $\confam$ be an arbitrary constraint family, and $\Sigma,\Sigma' \subseteq \{0,1\}$ be disjoint. We say that $\SAT_\Sigma(\confam)$ \emph{expresses} $\Sigma'$, if there is a constant $c$ such that the following holds: For any formula $\phi$ of $\SAT_{\Sigma\cup \Sigma'}(\confam)$ and parameter $k$, we can compute, in linear time, a formula $\phi'$ of $\SAT_\Sigma(\confam)$ with parameter $k':= k + c$ such that $\phi$ has a satisfying assignment of weight $k$ if and only if $\phi'$ has a satisfying assignment of weight $k'$.
\end{defn}

Indeed, for $0$-invalid $\confam$, we can show that $\SAT(\confam)$ expresses $\{0,1\}$ (this is straightforward and was already shown in~\cite{Marx05}). For $0$-valid $\confam$, however, expressing the constant $1$ in general appears impossible. To still give tight hardness results for $\confam$ whenever it represents a hard function $g$, we make use of a stronger notion that captures whether we can obtain $g$ already as a restriction that avoids the constant $1$.
Formally, let $f:\{0,1\}^r \to \{0,1\}, g:\{0,1\}^s\to \{0,1\}$ be arbitrary Boolean functions. We say that a function $f$ \emph{contains $g$ as a $0$-restriction} if $g$ is obtained from $f$ by replacing each argument of $f$ either by an argument of $g$ or the constant $0$, i.e., we can partition $[r]$ into $X_1,\dots, X_s, Z_0$ such that 
\[ g(x_1,\dots,x_s) = f(\overbrace{\entry{x_1}}^{X_1}, \dots, \overbrace{\entry{x_s}}^{X_s}, \overbrace{\entry{0}}^{Z_0}).\]

  Using careful constructions, we can prove the following central technical lemma.

\begin{lem}\label{lem:representingconstants}
Let $\confam$ be an arbitrary constraint family and let $g$ be $\Impl$ or $\NAND_d$ for some $d\ge 2$. If some $f\in \confam$ contains $g$ as a restriction, then $\SAT(\confam)$ expresses $\{0,1\}$, or $\SAT(\confam)$ expresses $0$ and $f$ contains $g$ already as a $0$-restriction.
\end{lem}

Postponing the proof of the above lemma to the Sections~\ref{sec:0invalid} and~\ref{sec:0valid}, we can give the proof of our hardness results.

\begin{thm}[Hardness for $\SAT(\confam)$]
Let $\confam$ be a constraint family.
\begin{enumerate}
\item If $\confam$ represents $\Impl$, then $\SAT(\confam)$ cannot be solved in time $f(k)\Oh(n^{(\omega/6-\eps)\sqrt[3]{k}+c})$ for any computable $f(k)$ and constants $c, \eps > 0$, unless the $k$-Clique conjecture fails.
\item If $\confam$ represents $\NAND_{2}$, then $\SAT(\confam)$ cannot be solved in time $f(k)\Oh(n^{(\omega/3-\eps)k+c})$ for any computable $f(k)$ and constants $c, \eps > 0$, unless the $k$-Clique conjecture fails. 
\item If $\confam$ represents $\NAND_{d}$ for some $d\ge 3$, then $\SAT(\confam)$ cannot be solved in time $f(k)\Oh(n^{(1-\eps)k+c})$ for any computable $f(k)$ and constants $c, \eps > 0$, unless the $d$-uniform HyperClique conjecture fails. 
\end{enumerate}
\end{thm}
\begin{proof}
First, we observe that \Implication reduces to $\SAT(\Impl)$ such that 
\begin{equation}\label{eq:impltoSAT}
T_\Implication(n,k)\le O(T_{\SAT(\Impl)}(n,k)).
\end{equation}
Indeed, given any directed graph $G=(V,E)$ with $V=\{v_1,\dots,v_n\}$, we define the formula $\phi$ with variables $x_1,\dots,x_n$ and the set of constraints obtained by including $x_i \Rightarrow x_j$ for all $(v_i,v_j)\in E$. Note that for any $S\subseteq [n]$, $\{v_i\}_{i\in S}$ is a valid set in $G$ iff $a_S$ is a satisfying assignment of $\phi$, yielding~\eqref{eq:impltoSAT}.

Similarly, we observe that the $d$-uniform HyperClique problem reduces to $\SAT(\NAND_d)$ such that 
\begin{equation}\label{eq:HCtoSAT}
\THC{d}(n,k)\le O(T_{\SAT(\NAND_d)}(n,k)).
\end{equation}
Indeed, given any $d$-uniform hypergraph $G=(V,E)$ with $|V|= n$, we define the formula $\phi$ with variables $x_1,\dots,x_n$ and the constraints obtained by including, for each distinct $v_{i_1},\dots,v_{i_d}\in V$ such that $(v_{i_1},\dots,v_{i_d}) \notin E$, the constraint $\NAND_d(x_{i_1}, \dots, x_{i_d})$. Observe that $(v_{i_1},\dots,v_{i_k})\in V^k$ is a hyperclique in $G$ iff the weight-$k$ assignment with $x_{i_\ell}=1$ for all $\ell \in [k]$ satisfies $\phi$, yielding~\eqref{eq:HCtoSAT}. 

It remains to show that whenever some $f\in \confam$ contains $g \in \{\Impl\} \cup \{\NAND_d \mid d\ge 2\}$ as a restriction, then there is a computable $f'(k)$ and constant $c'$ such that
\begin{equation}\label{eq:restToFam}
T_{\SAT(g)}(n,k) \le f'(k) \cdot T_{SAT(\confam)}(n,k+c').
\end{equation}
Indeed, if $\SAT(\confam)$ expresses $\{0,1\}$, then 
\[ T_{\SAT(g)}(n,k) \le \Oh(T_{\SAT_{\{0,1\}}(\confam)}(n,k)) \le f'(k) \Oh(T_{\SAT(\confam)}(n,k+c')). \]
Here the first inequality follows by replacing each occurrence of a constraint $g(x_{i_1},\dots,x_{i_d})$ of $\SAT(g)$ by the corresponding restriction $f(g_1(x_{i_1},\dots,x_{i_d}),\dots, g_r(x_{i_1},\dots,x_{i_d}))$ of $\SAT_{\{0,1\}}(\confam)$. The second inequality follows from the definition of $\SAT(\confam)$ expressing $\{0,1\}$.

In the other case, $\SAT(\confam)$ expresses only $0$, but $f$ contains $g$ already as a $0$-restriction. Then we have
\[ T_{\SAT(g)}(n,k) \le \Oh(T_{\SAT_{\{0\}}(\confam)}(n,k)) \le f'(k) \Oh(T_{\SAT(\confam)}(n,k+c')), \]
as replacing each occurrence of a constraint $g(x_{i_1},\dots,x_{i_d})$ of $\SAT(g)$ by the corresponding restriction $f(g_1(x_{i_1},\dots,x_{i_d}),\dots, g_r(x_{i_1},\dots,x_{i_d}))$ does not require the use of the constant $1$. The second inequality again follows from the definition of $\SAT(\confam)$ expressing $0$.

As a consequence, by \eqref{eq:impltoSAT} and \eqref{eq:restToFam}, a $f(k) \cdot O(n^{(\omega/6-\eps)\sqrt[3]{k}+c})$ $\SAT(\confam)$ algorithm for an $\Impl$-representing family $\confam$ would then give an \Implication algorithm running in time  
\[f(k) f'(k) \Oh(n^{(\omega/6-\eps)\sqrt[3]{k+c'}+c})= f''(k) \Oh(n^{(\omega/6-\eps)\sqrt[3]{k}+c''}),\]
where $f''(k) = f(k) f'(k)$ and $c'' \le c+\sqrt[3]{c'}$. This would refute the $k$-Clique conjecture by Theorem~\ref{thm:implhardness}, concluding 1.

Similarly, a $f(k) \cdot O(n^{\gamma k+c})$ $\SAT(\confam)$ algorithm for an $\NAND_d$-representing family $\confam$ would give a $d$-uniform HyperClique algorithm running in time  
\[f(k) f'(k) \Oh(n^{\gamma k+c+c'})= f''(k) \Oh(n^{\gamma k +c''}),\]
where $f''(k) = f(k) f'(k)$ and $c'' = c+c'$. This yields 2. and 3. by the $k$-Clique or $d$-uniform HyperClique conjecture, respectively.
\end{proof}

In the remainder of the section, we prove Lemma~\ref{lem:representingconstants}. We split the proof in two cases, depending on whether $f$ is 0-invalid (Lemma~\ref{lem:0invalidexpressesboth}) or 0-valid (Corollary~\ref{cor:0validfmainclaim}).

\subsection{Proof of Lemma~\ref{lem:representingconstants}: 0-invalid case}
\label{sec:0invalid}

Let $f$ be such that we can obtain $\Impl$ or $\NAND_d$ for $d\ge 2$ as a restriction. Note that if it contains $\NAND_{d},d>2$ then it also must contain $\NAND_2$ as a restriction. 

In this section, we consider the case that $f(y_1,\dots,y_r)$ is not $0$-valid, i.e., the all-zeroes assignment $u_1=\cdots = u_r= 0$ does not satisfy $f$.

\begin{lem}\label{lem:0invalidexpressesboth}
If $f$ contains $\Impl$ or $\NAND_2$ as a restriction and $f$ is 0-invalid, then $\SAT(\confam)$ expresses $\{0,1\}$.
\end{lem}

The above result in fact follows from the following claim.

\begin{claim}\label{claim:0invalidconst}
Let $f$ be as above. Given a parameter $k'$, we can compute, in time $\Oh(k')$, a formula $\phi_{0,1}$ of $\SAT(\confam)$ with variables $y, z_1, \dots, z_{k'+1}$ such that the only satisfying assignment of weight at most $k'$ is $y=1, z_1=\cdots = z_{k'+1} = 0$.
\end{claim}

Indeed, let us assume the above claim, and take any formula $\phi$ of $\SAT_{\{0,1\}}(\confam)$ with parameter~$k$. We construct $\phi_{0,1}$ with parameter $k':=k+1$ and define the formula $\phi'$ on variable set $x_1,\dots,x_n, y, z_1,\dots, z_{k'+1}$ where we include all constraints of $\phi_{0,1}$ and all constraints of $\phi$, \emph{replacing each use of the constant $0$ by $z_1$ and each use of the constant $1$ by $y$}. This yields a formula of $\SAT(\confam)$ with the property that for any weight-$k$ solution $x_1,\dots,x_n$ of $\phi$, the corresponding assignment that sets $y=1$ and $z_1=\cdots=z_{k'+1} = 0$ is a weight-$(k+1)$ solution of $\phi'$. Conversely, any $(k+1)$-weight solution of $\phi'$ must set $y=1$ and $z_1=0$ by the above claim, and hence the assignment to $x_1,\dots,x_n$ must also satisfy $\phi$. Observe that this proves Lemma~\ref{lem:0invalidexpressesboth}. 

\begin{proof}[Proof of Claim~\ref{claim:0invalidconst}]

We first give a set of constraints that enforces $y=1$. To this end, let $S\subseteq [r]$ be such that $a_S$ satisfies $f$; observe that $S$ exists and is non-empty (otherwise $f$ contains neither $\Impl$ nor $\NAND_{2}$ as a restriction). For each $j=1,\dots,k'+1$, define the constraint $C_j$ obtained by plugging in $y$ for each $u_i$ with $i\in S$ (i.e., all arguments set to $1$ under $a_S$), and $z_j$ for all other values. We claim that any weight-$(\le k')$ assignment satisfying $\bigwedge_{j=1}^{k'+1} C_j$ sets $y=1$: by the weight restriction, at least one of $z_{1},\dots,z_{k'+1}$ must be equal to $0$, say $z_{j^*}$. Then setting $y=0$ would falsify $C_{j^*}$, as then all its arguments are $0$. Note, however, that the desired assignment $y=1, z_1=\cdots=z_{k'+1}=0$ satisfies $\bigwedge_{j=1}^{k'+1} C_j$.

It remains to give additional constraints enforcing that $z_j = 0$ for \emph{all} $j\in [k'+1]$. As a first step, we find $S\subsetneq T$ such that $f(a_S)=1$ but $f(a_T)= 0$: Since $f$ represents $\Impl$ or $\NAND_2$, there is a partition of $[r]$ into $X, Y, Z_0, Z_1$ such that one of the following set of equalities hold:%
\[
\begin{matrix}
f( \overbrace{\entry{0}}^{X},& \overbrace{\entry{0}}^{Y}, & \overbrace{\entry{0}}^{Z_0}, & \overbrace{\entry{1}}^{Z_1}) & = & 1,\\
f( \entry{0},& \entry{1}, & \entry{0}, & \entry{1}) & = & 1,\\
f( \entry{1},& \entry{1}, & \entry{0}, & \entry{1}) & = & 1,\\
  \hline
f( \entry{1},& \entry{0}, & \entry{0}, & \entry{1}) & = & 0.
\end{matrix}
\qquad 
\begin{matrix}
f( \overbrace{\entry{0}}^{X},& \overbrace{\entry{0}}^{Y}, & \overbrace{\entry{0}}^{Z_0}, & \overbrace{\entry{1}}^{Z_1}) & = & 1,\\
f( \entry{1},& \entry{0}, & \entry{0}, & \entry{1}) & = & 1,\\
f( \entry{0},& \entry{1}, & \entry{0}, & \entry{1}) & = & 1,\\
  \hline
f( \entry{1},& \entry{1}, & \entry{0}, & \entry{1}) & = & 0.
\end{matrix}
\]
In both cases, the first and fourth line yield sets $S \subsetneq T$ with $f(a_S)=1$ and $f(a_T)=0$ (specifically, for $S=Z_1$ and $T=X \cup Z_1$ or for $S=Z_1$ and $T=X \cup Y \cup Z_1$).

Given such $S,T$, for each $j,j' \in {[r] \choose 2}$, we define the constraint $C'_{j,j'}$ obtained from $f(u_1,\dots,u_r)$ by plugging-in $y$ for all $u_i$ with $i\in S$, $z_j$ for all $i\in T\setminus S$ and $z_{j'}$ for all other $i$. Note that any satisfying assignment of weight at most $k$ sets at least one of $z_1,\dots, z_{k'+1}$ to $0$, say $z_{j^*}$. Observe that the constraint $C'_{j,j^*}$ is satisfied iff $z_j=0$, as setting $z_j$ to 0 or 1 corresponds to the assignments $a_S$ (satisfying) or $a_T$ (unsatisfying), respectively. Furthermore, observe that setting $y=1$ and $z_1=\cdots = z_{k'+1}=0$ indeed satisfies all $C'_{j,j'}$.
This concludes the claim that the only satisfying assignment of weight at most $k'$ is $y=1, z_1=\cdots = z_{k'+1} = 0$.
\end{proof}

\subsection{Proof of Lemma~\ref{lem:representingconstants}: 0-valid case}
\label{sec:0valid}

In this section, we consider the case that $f(y_1,\dots,y_r)$ is $0$-valid, i.e., the all-zeroes assignment $u_1=\cdots = u_r= 0$ satisfies $f$. We first observe that we can still express at least the constant 0.

\begin{lem}\label{lem:0validexpresses0}
If some $f\in \confam$ contains $\Impl$ or $\NAND_2$ as a restriction and $f$ is 0-valid, then $\SAT(\confam)$ expresses $0$.
\end{lem}
\begin{proof}
Observe that it suffices to show how to construct, given a parameter $k$, a formula on variables $z_1,\dots, z_{k+1}$ such that the only satisfying assignment of weight at most $k$ sets $z_1=\cdots = z_{k+1}=0$.

To this end, assume first that $f$ is not satisfied by the all-ones assignment. Then, the formula $\bigwedge_{i=1}^{k+1} f(z_i,\dots,z_i)$ is trivially only satisfied by the assignment $z_1=\cdots = z_{k+1} = 0$.

Otherwise, observe that there must be a non-empty set $S\subsetneq [r]$ such that $a_S$ does not satisfy $f$ (otherwise $f$ would be a trivial constraint and could contain neither of $\Impl$ and $\NAND_2$). For each $i,i'\in [k+1]$, we define the constraint $C_{i,i'}$ obtained by using $z_i$ for all arguments in $S$, and $z_{i'}$ for all arguments not in $S$. Observe that $C_{i,i'} \wedge C_{i',i}$ forces $z_i = z_{i'}$, and thus $z_1= \cdots = z_{k+1}$, which is satisfied by an assignment of weight at most $k$ if and only if the common value is $0$. 
\end{proof}

Interestingly, for $0$-valid $f$, containing $\Impl$ as a restriction is equivalent to containing $\Impl$ already as a $0$-restriction.

\begin{lem}\label{lem:Implas0restr}
If $f$ contains $\Impl$ as a restriction and is $0$-valid, then $f$ contains $\Impl$ already as a $0$-restriction.
\end{lem}
\begin{proof}
Since $f:\{0,1\}^r \to \{0,1\}$ contains $\Impl$ as a restriction, we can partition $[r]$ into sets $X,Y,Z_0,Z_1$ and write
\begin{align}\label{eq:containsImpl}
\begin{matrix}
f( \overbrace{\entry{0}}^{X},& \overbrace{\entry{0}}^{Y}, & \overbrace{\entry{0}}^{Z_0}, & \overbrace{\entry{1}}^{Z_1}) & = & 1,\\
f( \entry{0},& \entry{1}, & \entry{0}, & \entry{1}) & = & 1,\\
f( \entry{1},& \entry{1}, & \entry{0}, & \entry{1}) & = & 1,\\
  \hline
f( \entry{1},& \entry{0}, & \entry{0}, & \entry{1}) & = & 0.\\
\end{matrix}
\end{align}
Assume first that 
\begin{equation}\label{eq:casedist}
f(\overbrace{\entry{0}}^{X}, \overbrace{\entry{1}}^Y, \overbrace{\entry{0}}^{Z_0}, \overbrace{\entry{0}}^{Z_1}) = 0.
\end{equation}
Then, we obtain $\Impl$ as a 0-restriction by setting $X'\coloneqq Y, Y'\coloneqq Z_1, Z' \coloneqq X\cup Z_0$ and observing that
\begin{align*}
\begin{matrix}
f( \overbrace{\entry{0}}^{X'=Y},& \overbrace{\entry{0}}^{Y'=Z_1}, & \overbrace{\entry{0}}^{Z'=X\cup Z_0}) & = & 1, & \quad \text{[$f$ is $0$-valid]}\\
f( \entry{0},& \entry{1}, & \entry{0}) & = & 1, & \quad\text{[by \eqref{eq:containsImpl}]}\\
f( \entry{1},& \entry{1}, & \entry{0}) & = & 1, & \quad\text{[by \eqref{eq:containsImpl}]}\\
  \hline
f( \entry{1},& \entry{0}, & \entry{0}) & = & 0. & \quad\text{[by \eqref{eq:casedist}]}\\
\end{matrix}
\end{align*}
Otherwise, if \eqref{eq:casedist} does not hold, then we obtain $\Impl$ as a 0-restriction by setting $X'\coloneqq X\cup Z_1, Y' \coloneqq Y, Z' \coloneqq Z_0$ and observing that 
\begin{align*}
\begin{matrix}
f( \overbrace{\entry{0}}^{X'=X\cup Z_1},& \overbrace{\entry{0}}^{Y'=Y}, & \overbrace{\entry{0}}^{Z'=Z_0}) & = & 1, & \quad \text{[$f$ is $0$-valid]}\\
f( \entry{0},& \entry{1}, & \entry{0}) & = & 1, & \quad\text{[by $\neg$\eqref{eq:casedist}]}\\
f( \entry{1},& \entry{1}, & \entry{0}) & = & 1, & \quad\text{[by \eqref{eq:containsImpl}]}\\
  \hline
f( \entry{1},& \entry{0}, & \entry{0}) & = & 0. & \quad\text{[by \eqref{eq:containsImpl}]}\\
\end{matrix}
\end{align*}
\end{proof}

It remains to handle the case that $f$ contains $\NAND_d$ as a restriction.
We first observe that if $f$ contains $\Impl$ as a $0$-restriction, then $\SAT_0(\confam)$ even expresses the constant $1$. (Thus, afterwards, we may assume that $f$ does not contain $\Impl$ as a 0-restriction.)

\begin{lem}\label{lem:0validexpresses1}
If some $f\in \confam$ contains $\Impl$ as a $0$-restriction, then $\SAT_0(\confam)$ expresses $1$. 
\end{lem}
\begin{proof}
Given any formula $\phi$ of $\SAT_{\{0,1\}}(\confam)$ on variables $x_1,\dots,x_n$, we construct a formula $\phi'$ on variables $x_1,\dots, x_n, y$ as follows: Since some $f\in \confam$ contains $\Impl$ as a $0$-restriction, we can express, for any variables $v,v'$, the implication $v \Rightarrow v'$ by a corresponding constraint of $\SAT_0(\confam)$. We construct $n$ such constraints to enforce $\bigwedge_{j=1}^{n} (x_j\Rightarrow  y)$. Subsequently, we may use $y$ to replace any use of the constant $1$ to convert the constraints of $\phi$ to constraints of the $\SAT_{0}(\confam)$-formula $\phi'$. 

To argue correctness, note that any satisfying weight-$k$ assignment of $\phi$ yields a satisfying weight-$(k+1)$ assignment of $\phi'$ by setting $y=1$. Conversely, note that any weight-$(k+1)$-assignment of $\phi'$ must set $y=1$ (since $k\ge 1$ implies that at least one variable $x_i$ is set to one, which enforces $y=1$ by the corresponding implication $x_i \Rightarrow y$) and thus corresponds to a weight-$k$ assignment to $x_1,\dots, x_n$ satisfying $\phi$. 
\end{proof}

In the remainder of this section, we assume that $f$ contains $\NAND_d$ as a restriction, but does not contain $\Impl$ as a $0$-restriction, and the aim is to find $\NAND_d$ already as a $0$-restriction.

\begin{lem}\label{lem:subsetdiff}
For any $0$-valid $f$, if $f$ does not contain $\Impl$ as a 0-restriction, then whenever $f(a_S) = f(a_T) = 1$ with $S\subseteq T$, then $f(a_{T\setminus S}) = 1$.
\end{lem}
\begin{proof}
If $S=T$, there is nothing to show, so let $S\subsetneq T$ and assume for contradiction that $f(a_{T\setminus S}) = 0$. We obtain $\Impl$ as a 0-restriction as follows:
\begin{align*}
\begin{matrix}
f( \overbrace{\entry{0}}^{X=T\setminus S},& \overbrace{\entry{0}}^{Y=S}, & \overbrace{\entry{0}}^{Z=[r]\setminus T}) & = & 1, & \quad \text{[$f$ is $0$-valid]}\\
f( \entry{0},& \entry{1}, & \entry{0}) & = & 1, & \quad\text{[$f(a_S)=1$]}\\
f( \entry{1},& \entry{1}, & \entry{0}) & = & 1, & \quad\text{[$f(a_T)=1$]}\\
  \hline
f( \entry{1},& \entry{0}, & \entry{0}) & = & 0. & \quad\text{[by assumption]}\\
\end{matrix}
\end{align*}
This yields the claim.
\end{proof}

We can finally obtain $\NAND_d$ as a $0$-restriction.

\begin{lem}\label{lem:NANDas0restr}
If $f$ contains $\NAND_d$ as a restriction, does not contain $\Impl$ as a 0-restriction and is $0$-valid, then $f$ contains $\NAND_d$ already as a $0$-restriction.
\end{lem}
\begin{proof}
Since $f:\{0,1\}^r \to \{0,1\}$ contains $\NAND_d$ as a restriction, we can partition $[r]$ into sets $X_1,\dots,X_d,Z_0,Z_1$ such that $X_I\coloneqq \bigcup_{i\in I} X_i$ with $I\subseteq [d]$ satisfies:
\begin{equation}\label{eq:NANDrestr}
f(a_{X_I \cup Z_1}) = \begin{cases} 0 & \text{if } I = [d], \\ 1 & \text{if } I \subsetneq [d]. \end{cases}
\end{equation}

We claim that the partition $X'_{i} \coloneqq X_i$ for $i<d$, $X'_d \coloneqq X_d \cup Z_1$, $Z'\coloneqq Z_0$ provides $\NAND_d$ as a 0-restriction: Letting $X'_I \coloneqq \bigcup_{i\in I} X'_i$, this follows from
\begin{equation}\label{eq:NAND0restr}
f(a_{X'_I}) = \begin{cases} 0 & \text{if } I = [d], \\ 1 & \text{if } I \subsetneq [d]. \end{cases}
\end{equation}
To verify~\eqref{eq:NAND0restr}, note first that $f(a_{X'_{[d]}}) = f(a_{X_{[d]}\cup Z_1}) = 0$ by~\eqref{eq:NANDrestr}. Second, let $I\subsetneq [d]$. If $d\in I$, then $f(a_{X'_I}) = f(a_{X_I\cup Z_1}) = 1$ by~\eqref{eq:NANDrestr}. Otherwise, if $d\notin I$, then we have $f(a_{X'_I}) = f(a_{X_I}) = 1$ by Lemma~\ref{lem:subsetdiff} (for this, note that $f$ does not contain $\Impl$ as 0-restriction and that $f(a_{X_I \cup Z_1}) = f(a_{Z_1}) = 1$). This concludes the claim.   
\end{proof}

The proof of this section is summarized in the following corollary. 

\begin{cor}\label{cor:0validfmainclaim}
If $f$ contains $g\in \{\Impl\} \cup \bigcup_{d\ge 2} \{\NAND_d\}$ and $f$ is $0$-valid, then $\SAT(f)$ expresses $\{0,1\}$, or $\SAT(f)$ expresses $0$ and contains $g$ as a $0$-restriction. 
\end{cor}
\begin{proof}
If $g=\Impl$, then $f$ contains $g$ already as a $0$-restriction by Lemma~\ref{lem:Implas0restr} and $\SAT(f)$ expresses $\{0,1\}$ by Lemmas~\ref{lem:0validexpresses0} and~\ref{lem:0validexpresses1}.

If $g=\NAND_d$, then either $f$ also contains $\Impl$ as a $0$-restriction, in which case $\SAT(f)$ expresses $\{0,1\}$ by Lemmas~\ref{lem:0validexpresses0} and~\ref{lem:0validexpresses1}, or it does not contain $\Impl$ as a $0$-restriction, and thus $f$ contains $g$ as a $0$-restriction by Lemma~\ref{lem:NANDas0restr} and $\SAT(f)$ expresses $0$ by Lemma~\ref{lem:0validexpresses0}.
\end{proof}

\bibliographystyle{plainurl}%
{\bibliography{weightedSAT}}

\end{document}